\documentclass[a4paper,UKenglish]{lipics-v2018}

\usepackage{microtype}

\usepackage{xcolor}
\usepackage{hyperref}
\usepackage{xspace}
\usepackage{booktabs}
\usepackage{threeparttable}
\usepackage{tabularx}
\usepackage{amsmath}
\usepackage{amsthm}
\usepackage{amssymb}
\newcommand{\Oh}{\mathcal{O}}
\newcommand{\OR}{\textsc{or}\xspace}
\newcommand{\notcontainment}{\ensuremath{\mathsf{NP \not\subseteq coNP/poly}}\xspace}

\newcommand{\eqvr}[0]{\ensuremath{\mathcal{R}}\xspace}
\newcommand{\ERBDS}{\textsc{erbds}\xspace}
\newcommand{\SERBDS}{\textsc{semi-erbds}\xspace}
\newcommand{\RBDS}{\textsc{rbds}\xspace}
\newcommand{\containment}{\ensuremath{\mathsf{NP  \subseteq coNP/poly}}\xspace}
\newcommand{\ncontainment}{\ensuremath{\mathsf{NP  \not \subseteq coNP/poly}}\xspace}

\newcommand{\yes}{\textsc{yes}}
\newcommand*{\etal}{\textit{et al.}\xspace}

\newcommand{\vect}[1]{\mathbf{#1}}
\newcommand{\Q}{\ensuremath{\mathcal{Q}}}
\newcommand{\N}{\ensuremath{\mathbb{N}}}

\definecolor{gray}{RGB}{140, 140, 140}
\definecolor{blue2}{RGB}{0, 20, 120}
\newcommand{\problem}[1]{\textsc{#1}}
\newcommand{\todo}[1][]{%
  \ifx/#1/%
    \textcolor{red}{TODO!}%
  \else%
    \textcolor{red}{todo: #1}%
  \fi%
}
\theoremstyle{plain}
\newtheorem{claim}[theorem]{Claim}
\newtheorem{observation}[theorem]{Observation}
\let\plainqed\qedsymbol
\newcommand{\claimqed}{$\lrcorner$}

\newenvironment{claimproof}{\begin{proof}\renewcommand{\qedsymbol}{\claimqed}}{\end{proof}\renewcommand{\qedsymbol}{\plainqed}}
\newenvironment{note}{%
}{%
}%

\newcommand{\defparproblem}[4]{
 \vspace{1mm}
\noindent\fbox{
 \begin{minipage}{0.96\textwidth}
 \begin{tabular*}{\textwidth}{@{\extracolsep{\fill}}lr} \problem{#1} & {\bf{Parameter:}} #3 \\ \end{tabular*}
 {\bf{Input:}} #2 \\
 {\bf{Question:}} #4
 \end{minipage}
 }
 \vspace{1mm}
}

\newcommand{\defproblem}[3]{
 \vspace{1mm}
\noindent\fbox{
 \begin{minipage}{0.96\textwidth}
 \begin{tabular*}{\textwidth}{@{\extracolsep{\fill}}lr} \problem{#1} &  \\ \end{tabular*}
 {\bf{Input:}} #2 \\
 {\bf{Question:}} #3
 \end{minipage}
 }
 \vspace{1mm}
}

\newcommand{\dSAT}{$d$-\textsc{cnf-sat}\xspace}
\newcommand{\dNAESAT}{$d$-\textsc{nae-sat}\xspace}
\newcommand{\ExactdSAT}{\textsc{Exact $d$-sat}\xspace}
\newcommand{\ExactSAT}{\textsc{Exact sat}\xspace}
\newcommand{\Ohtilde}{\ensuremath{\widetilde{\mathcal{O}}}}
\newcommand{\rootCSP}{\textsc{$d$-Polynomial root CSP}\xspace}
\newcommand{\nonrootCSP}{\textsc{$d$-Polynomial non-root CSP}\xspace}
\newcommand{\linearNonrootCSP}{\textsc{$1$-Polynomial non-root CSP}\xspace}
\newcommand{\linearrootCSP}{\textsc{$1$-Polynomial root CSP}\xspace}
\newcommand{\primefield}{\ensuremath{\mathbb{Z}/p\mathbb{Z}}\xspace}
\newcommand{\true}{\emph{true}\xspace}
\newcommand{\false}{\emph{false}\xspace}
\newcommand{\A}[4]{\begin{note}\ensuremath{a^{#1,#2}_{#3,#4}}\end{note}}
\newcommand{\B}[3]{\begin{note}\ensuremath{b^{#1}_{#2,#3}}\end{note}}
\nolinenumbers
\hideLIPIcs

\title{Optimal Sparsification for Some Binary CSPs Using Low-Degree Polynomials\footnote{An extended abstract of this work appeared under the same title in the \emph{Proceedings of the 41st International Symposium on Mathematical Foundations of Computer Science (MFCS), Krakow, Poland, August 2016.}}}

\author{Bart M.\,P. Jansen}{Eindhoven University of Technology\\{P.O. Box 513, 5600 MB Eindhoven, The Netherlands}}{b.m.p.jansen@tue.nl}{http://orcid.org/0000-0001-8204-1268}{}

\author{Astrid Pieterse}{Eindhoven University of Technology\\{P.O. Box 513, 5600 MB Eindhoven, The Netherlands}}{astridpieterse@outlook.com}{http://orcid.org/0000-0003-3721-6721}{}

\authorrunning{B.\,M.\,P. Jansen and A. Pieterse}

\keywords{constraint satisfaction problem, kernelization, satisfiability, sparsification}

\funding{
This work was supported by NWO Veni grant ``Frontiers in Parameterized Preprocessing'' and NWO
Gravitation grant ``Networks''. This work was done in part while the first author was visiting the Simons Institute for the Theory of Computing.
}

\begin{document}
\maketitle

\begin{abstract}
This paper analyzes to what extent it is possible to efficiently reduce the number of clauses in NP-hard satisfiability problems, without changing the answer. Upper and lower bounds are established using the concept of kernelization. Existing results show that if \ncontainment, no efficient preprocessing algorithm can reduce $n$-variable instances of \textsc{cnf-sat} with~$d$ literals per clause, to equivalent instances with~$\Oh(n^{d-\varepsilon})$ bits for any~$\varepsilon > 0$. For the \textsc{Not-All-Equal sat} problem, a compression to size~$\Ohtilde(n^{d-1})$ exists. We put these results in a common framework by analyzing the compressibility of binary CSPs. We characterize constraint types based on the minimum degree of multivariate polynomials whose roots correspond to the satisfying assignments, obtaining (nearly) matching upper and lower bounds in several settings. Our lower bounds show that not just the number of constraints, but also the encoding size of individual constraints plays an important role. For example, for \textsc{Exact Satisfiability} with unbounded clause length it is possible to efficiently reduce the number of constraints to~$n+1$, yet no polynomial-time algorithm can reduce to an equivalent instance with~$\Oh(n^{2-\varepsilon})$ bits for any~$\varepsilon > 0$, unless \containment.
\end{abstract}



\maketitle
\section{Introduction}

The goal of sparsification is to make an object such as a graph or logical structure less dense, without changing the outcome of a computational task of interest. Sparsification can be used to speed up the solution of NP-hard problems, by sparsifying a problem instance before solving it. The notion of kernelization, originating in the field of parameterized complexity~\cite{CyganFKLMPPS15,DowneyF13,FlumG06}, facilitates a rigorous study of polynomial-time preprocessing for NP-hard problems and can be used to reason about (the impossibility of) sparsification. Over the last few years, our understanding of the power of polynomial-time data reduction has increased tremendously, as documented in recent surveys~\cite{Bodlaender15,Gutin15a,Kratsch14,LokshtanovMS12}. By studying the kernelization complexity of a graph problem parameterized by the number of vertices, or of a logic problem parameterized by the number of variables, we can analyze its potential for sparsification.

The vast majority of the currently known results in this direction are negative~\cite{SatisfiabilityDell14,Jansen15,Jansen16,sparsificationJansenP15}, stating that no nontrivial sparsification is possible under plausible complexity-theoretic assumptions. For example, Dell and van Melkebeek~\cite{SatisfiabilityDell14} obtained such a result for \textsc{CNF-Satisfiability} with clauses of size at most~$d$ (\dSAT), for each fixed~$d \geq 3$. Assuming \ncontainment, there is no polynomial-time algorithm that compresses any $n$-variable instance of \dSAT to an equivalent instance with~$\Oh(n^{d-\varepsilon})$ bits for~$\varepsilon > 0$. Since there are~$\Oh(n^d)$ possible clauses of size at most~$d$ over~$n$ variables, the trivial compression scheme that outputs a bitstring of length~$\Oh(n^d)$, denoting for each possible clause whether it occurs in the instance or not, is optimal up to~$n^{o(1)}$ factors.

A problem for which nontrivial polynomial-time sparsification \emph{is} possible was recently discovered by the current authors~\cite{sparsificationJansenP15}. Any $n$-variable instance of the \textsc{Not-All-Equal CNF-Satisfiability} problem with clauses of size at most~$d$ (henceforth called \dNAESAT) can efficiently be compressed to an equivalent instance with~$\Oh(n^{d-1})$ clauses, which can be encoded in~$\Oh(n^{d-1} \log n)$ bits. The preprocessing algorithm is based on a linear-algebraic lemma by Lov\'asz~\cite{Lovasz76} to identify clauses that are implied by others, allowing a reduction from~$\Theta(n^d)$ clauses to~$\Oh(n^{d-1})$. This sparsification for \dNAESAT forms the starting point for this work. Since \dSAT and \dNAESAT can both be seen as constraint satisfaction problems (CSPs) with a binary domain, it is natural to ask whether the positive results for \dNAESAT extend to other binary CSPs. The difference between \dSAT and \dNAESAT shows that the type of constraints that one allows, affects the compressibility of the resulting CSP. The goal of this paper is to understand how the optimal compression size for a binary CSP depends on the type of legal constraints, with the aim of obtaining matching upper and lower bounds.

Before presenting our results, we give an example to illustrate our methods. Consider the NP-complete \textsc{Exact $d$-CNF-Satisfiability} (\ExactdSAT) problem, which asks whether there is a truth assignment that satisfies \emph{exactly one} literal in each clause; the clauses have size at most~$d$. While there are~$\Theta(n^d)$ different clauses that can occur in an instance with~$n$ variables, the exact nature of the problem makes it possible to reduce any instance to an equivalent one with~$n+1$ clauses. A clause such as~$x_1 \vee x_3 \vee \neg x_5$ naturally corresponds to an equality constraint of the form~$x_1 + x_3 + (1 - x_5) = 1$, since a $0/1$-assignment to the variables satisfies exactly one literal of the clause if and only if it satisfies the equality. To find redundant clauses, transform each of the~$m$ clauses into an equality to obtain a system of equalities~$A\vect{x} = \vect{b}$ where~$A$ is an~$m \times n$ matrix, $\vect{x}$ is the column vector~$(x_1, \ldots, x_n)$, and~$\vect{b}$ is an integer column vector. Using Gaussian elimination, one can efficiently compute a basis~$B$ for the row space of the extended matrix~$(A|b)$: a set of equalities such that every equality can be written as a linear combination of equalities in~$B$. Since~$(A|b)$ has~$n+1$ columns, its rank is at most~$n+1$ and the basis~$B$ contains at most~$n+1$ equalities. To perform data reduction, remove all clauses from the \ExactdSAT instance whose corresponding equalities do not occur in~$B$. If an assignment satisfies~$f_1(\vect{x}) = b_1$ and~$f_2(\vect{x}) = b_2$, then it also satisfies their sum $f_1(\vect{x}) + f_2(\vect{x}) = b_1 + b_2$, and any linear combination of the satisfied equalities. Since any equality not in~$B$ can be written as a linear combination of equalities in~$B$, a truth assignment satisfying all clauses from~$B$ must necessarily also satisfy the remaining clauses, which shows the correctness of the data reduction procedure. The resulting instance can be encoded in~$\Oh(n \log n)$ bits, as each of the remaining~$n+1$ clauses has~$d \in O(1)$ literals.

\subparagraph*{Our results} Our positive results are generalizations of the linear-algebraic data reduction tool for binary CSPs presented above. They reveal that the~$\Ohtilde(n)$-bit compression for \ExactdSAT, the~$\Ohtilde(n^{d-1})$-bit compression for \dNAESAT, and the~$\Oh(n^d)$-bit compression for \dSAT are samples of a gliding scale of problem complexity: more tightly constrained problems can be compressed better. We formalize this idea by considering a generic CSP whose constraints are of the form $f(\vect{x}) = 0$, where~$f$ is a bounded-degree multivariate polynomial and the constraint demands that~$\vect{x}$ is a root of~$f$. The example given earlier shows that \ExactdSAT can be expressed using degree-1 polynomials. We show that \dNAESAT and \dSAT can be expressed using equalities of polynomial expressions of degree~$d-1$ and~$d$. We therefore study the following problem:

\defparproblem{\problem{\rootCSP}}
{A list $L$ of polynomial equalities over variables $V = \{x_1, \ldots, x_n\}$. An equality is of the form $f(x_1, \ldots, x_n) = 0$, where $f$ is a multivariate polynomial of degree at most~$d$.}
{The number of variables $n$.}
{Does there exist an assignment of the variables $\tau \colon V \to \{0,1\}$ satisfying all equalities in $L$?
}

Using a generalization of the argument presented above, the number of constraints in an instance of \rootCSP can efficiently be reduced to~$\Oh(n^d)$, even when the number of variables that occur in a constraint is not restricted. The latter implies, for example, that using degree-1 polynomials one can express the \ExactSAT problem with clauses of arbitrary size. When the number of variable occurrences in a constraint can be as large as~$n$, it may take~$\Omega(n)$ bits to encode a single constraint. After reducing the number of clauses in an \ExactSAT instance to~$n+1$, one may therefore still require~$\Theta(n^2)$ bits to encode the instance. This turns out to be unavoidable: we prove that \ExactSAT has no sparsification of size~$\Oh(n^{2-\varepsilon})$ for any~$\varepsilon > 0$, unless \containment. In general, we compress instances of~\rootCSP to bitsize~$\Ohtilde(n^{d+1})$ when each constraint can be encoded in~$\Ohtilde(n)$ bits. We prove that no compression to size~$\Oh(n^{d+1-\varepsilon})$ is possible unless \containment. When each constraint can be encoded in~$\Ohtilde(1)$ bits, the constraint reduction scheme reduces the size of an instance to~$\Ohtilde(n^d)$. As we will show that \dNAESAT can be modeled using polynomials of degree~$d-1$, this method strictly generalizes our earlier results~\cite{sparsificationJansenP15} for \dNAESAT.

The linear-algebraic data reduction tool described above works over arbitrary fields~$F$, allowing us to capture constraints such as ``the number of satisfied literals in the clause is exactly two, when evaluated modulo~$3$''. We therefore extend our study to the \rootCSP problem over arbitrary fields~$F$, and obtain similar positive and negative results. We furthermore extend our previous work by showing similar upper and lower bounds for \rootCSP over the integers modulo $m$, where $m$ need not be a prime number. When $m$ is not prime, the resulting structure is not a field, which imposes technical difficulties.

Finally, we consider binary CSPs whose constraints are formed by \emph{inequalities}, rather than equalities, of degree-$d$ polynomials. This leads to the following generic problem:

\defparproblem{\nonrootCSP over $F$}
{A list $L$ of polynomial inequalities over variables $V = \{x_1, \ldots, x_n\}$. An inequality is of the form $f(x_1, \ldots, x_n) \neq 0$, where $f$ is a multivariate polynomial of degree at most~$d$.}
{The number of variables $n$.}
{Does there exist an assignment of the variables $\tau \colon V \to \{0,1\}$ satisfying all inequalities in $L$?}

We present upper and lower bounds for problems of this type. When the polynomials are evaluated over a structure that is not a field, the behavior changes significantly. For example, CSPs with constraints of the type ``the number of satisfied literals in the clause is 1 or 2, when evaluated modulo 6'' behave differently than the corresponding problem modulo 5, or modulo 7, because the integers modulo 6 do not form a field.
In contrast to \rootCSP, our lower-bound techniques for \nonrootCSP fail when defining constraints with respect to composite moduli.
We present connections to different areas of theoretical computer science where the distinction between prime and composite moduli plays a big role. More concretely, we show that obtaining polynomial sparsification upper bounds for \nonrootCSP over the integers modulo a composite is strongly tied to long-standing open problems concerning the representation of the \OR-function using low-degree polynomials (cf.~\cite{Barrington1994representing,Bhowmick2015Nonclassical,TardosB98}). Table \ref{table:results} contains a summary of our results.
\begin{table}
\begin{threeparttable}
\caption{Summary of the kernel upper and lower bounds obtained in this paper, expressed in the number of bits. The bounds depend on whether the polynomials defining the constraints are over the rationals~$\mathbb{Q}$, the integers modulo a prime~$p$, or the integers modulo a composite~$m$. The integer $r$ denotes the number of distinct prime divisors of $m$. The values of~$p$,~$m$,~$r$, and~$d$ are treated as constants in these bounds. \label{table:results}}{
\begin{tabularx}{\textwidth}{@{}r X X c l X@{}}
\toprule
Problem                 &             \multicolumn{2}{c}{\rootCSP} && \multicolumn{2}{c}{\nonrootCSP} \\
\cmidrule{2-3}\cmidrule{5-6}
                        &           Lower bound{$^1$} & Upper bound{$^2$} && Lower bound{$^1$} & Upper bound{$^2$} \\
\midrule
$\mathbb{Q}$            &$\Omega(n^{d+1-\varepsilon})$&$\widetilde{\Oh}(n^{d+1})$
&& Superpolynomial &  \\
$\mathbb{Z}/p\mathbb{Z}$&$\Omega(n^{d+1-\varepsilon})$&$\widetilde{\Oh}(n^{d+1})$
&& $\Omega(n^{d(p-1)-\varepsilon})$ & $\widetilde{\Oh}(n^{d(p-1)+1})$ \\
$\mathbb{Z}/m\mathbb{Z}$&$ \Omega(n^{d+1-\varepsilon})$&$\widetilde{\Oh}(n^{d+1})$&& $\Omega(n^{(d/2)^r-\varepsilon})$  & ? \\
\midrule
Reference & Thm. \ref{cor:LB:1-poly-root-csp-modm}, \ref{thm:drootcspmod:lb} & Thm. \ref{thm:kernel}, \ref{thm:kernel_modm}, \ref{thm:subset_kernel_modm} &  & Thm. \ref{thm:nonroot:q:lb}, \ref{thm:nonroot:prime:lb},  \ref{thm:nonroot:composite:lb} & Thm. \ref{thm:UB:non-root}\\
\bottomrule
\end{tabularx}}
\begin{tablenotes}
\item [1] \small{The lower bounds hold for any $\varepsilon > 0$, for the problems that are not polynomial-time solvable and under the assumption that \notcontainment.}
\item [2] \small{The upper bounds hold when each $n$-variate polynomial constraint in the input can be encoded in $\widetilde{\Oh}(n)$ bits.}
\end{tablenotes}
\end{threeparttable}
\end{table}

\subparagraph*{Related work} Schaefer's Theorem~\cite{Schaefer78} is a classic result relating the complexity of a binary CSP to the type of allowed constraints, separating the NP-complete from the polynomial-time solvable cases. A characterization of the kernelization complexity of min-ones CSPs parameterized by the number of variables was presented by Kratsch and Wahlstr\"om~\cite{Kratsch2010preprocessing}. There are several parameterized complexity results for CSPs~\cite{ConstraintBulatov14,ComplexityDell15,KratschMW16}.

\section{Preliminaries} \label{sec:preliminaries}

A \emph{parameterized problem} $\mathcal{Q}$ is a subset of $\Sigma^* \times \mathbb{N}$, where $\Sigma$ is a finite alphabet. Let $\mathcal{Q}, \mathcal{Q}' \subseteq \Sigma^*\times\mathbb{N}$ be parameterized problems and let $h\colon\mathbb{N}\rightarrow\mathbb{N}$  be a computable function. A \emph{generalized kernel for $\mathcal{Q}$ into $\mathcal{Q}'$ of size $h(k)$} is an algorithm that, on input $(x,k) \in \Sigma^*\times\mathbb{N}$, takes time polynomial in $|x|+k$ and outputs an instance $(x',k')$ such that:
\begin{enumerate}
\item $|x'|$ and $k'$ are bounded by $h(k)$, and
\item $(x',k')\in\mathcal{Q}'$ if and only if $(x,k) \in \mathcal{Q}$.
\end{enumerate}
The algorithm is a \emph{kernel} for $\mathcal{Q}$ if $\mathcal{Q} = \mathcal{Q'}$. It is a \emph{polynomial (generalized) kernel} if $h(k)$ is a polynomial. Since a polynomial-time reduction to an equivalent sparse instance yields a generalized kernel, we use lower bounds for the sizes of generalized kernels to prove the non-existence of sparsification algorithms.

A \emph{linear-parameter transformation} from a parameterized problem~$\mathcal{Q}$ to a parameterized problem~$\mathcal{Q'}$ is a polynomial-time algorithm that transforms any instance~$(x,k)$ of~$\mathcal{Q}$ into an equivalent instance~$(x',k')$ of~$\mathcal{Q'}$ such that~$k' \in \Oh(k)$. It is easy to see (cf.~\cite{BodlaenderTY11}) that the existence of a linear-parameter transformation from~$\mathcal{Q}$ to~$\mathcal{Q'}$, together with a (generalized) kernel of size~$\Oh(k^d)$ for~$\mathcal{Q'}$, yields a generalized kernel of size~$\Oh(k^d)$ for~$\mathcal{Q}$. By contraposition, the existence of such a transformation implies that when~$\mathcal{Q}$ does not have generalized kernels of size~$\Oh(k^{d-\varepsilon})$, then~$\mathcal{Q'}$ does not have generalized kernels of size~$\Oh(k^{d-\varepsilon})$ either. For some of our lower bounds, we use linear-parameter transformations in combination with the following result by Dell and van Melkebeek. They proved a stronger version of the following theorem in \cite{SatisfiabilityDell14}. It is rephrased here to match our terminology.

\begin{theorem}[{\cite[Theorem 1]{SatisfiabilityDell14}}]\label{thm:lower_bound:CNF}
Let $d \geq 3$ be an integer. Then \dSAT parameterized by the number of variables~$n$ does not have a generalized kernel of size $\Oh(n^{d-\varepsilon})$ for any $\varepsilon > 0$, unless \containment.
\end{theorem}

We also use the framework of cross-composition~\cite{BodlaenderJK14} to establish kernelization lower bounds, requiring the definitions of polynomial equivalence relations and \OR-cross-compositions.

\begin{definition}[{Polynomial equivalence relation, \cite[Def. 3.1]{BodlaenderJK14}}] \label{definition:eqvr}
An equivalence relation~\eqvr on~$\Sigma^*$ is called a \emph{polynomial equivalence relation} if the following conditions hold.
\begin{itemize}
\item There is an algorithm that, given two strings~$x,y \in \Sigma^*$, decides whether~$x$ and~$y$ belong to the same equivalence class in time polynomial in~$|x| + |y|$.
\item For any finite set~$S \subseteq \Sigma^*$ the equivalence relation~$\eqvr$ partitions the elements of~$S$ into a number of classes that is polynomially bounded in the size of the largest element of~$S$.
\end{itemize}
\end{definition}

\begin{definition}[{Cross-composition, \cite[Def. 3.3]{BodlaenderJK14}}]\label{definition:crosscomposition}
Let~$L\subseteq\Sigma^*$ be a language, let~$\eqvr$ be a polynomial equivalence relation on~$\Sigma^*$, let~$\Q\subseteq\Sigma^*\times\N$ be a parameterized problem, and let~$f \colon \N \to \N$ be a function. An \emph{\OR-cross-com\-position of~$L$ into~$\Q$} (with respect to \eqvr) \emph{of cost~$f(t)$} is an algorithm that, given~$t$ instances~$x_1, x_2, \ldots, x_t \in \Sigma^*$ of~$L$ belonging to the same equivalence class of~$\eqvr$, takes time polynomial in~$\sum _{i=1}^t |x_i|$ and outputs an instance~$(y,k) \in \Sigma^* \times \mathbb{N}$ such that:
\begin{itemize}
\item The parameter~$k$ is bounded by $\Oh(f(t)\cdot(\max_i|x_i|)^c)$, where~$c$ is some constant independent of~$t$, and
\item instance $(y,k) \in \Q$ if and only if there is an~$i \in [t]$ such that~$x_i \in L$.\label{property:OR}
\end{itemize}
\end{definition}

\begin{theorem}[{\cite[Theorem 6]{BodlaenderJK14}}] \label{thm:cross_composition_LB}
Let~$L\subseteq\Sigma^*$ be a language, let~$\Q\subseteq\Sigma^*\times\N$ be a parameterized problem, and let~$d,\varepsilon$ be positive reals. If~$L$ is NP-hard under Karp reductions, has an \OR-cross-composition into~$\Q$ with cost~$f(t)=t^{1/d+o(1)}$, where~$t$ denotes the number of instances, and~$\Q$ has a polynomial (generalized) kernelization with size bound~$\Oh(k^{d-\varepsilon})$, then \containment.
\end{theorem}

For~$d \in \N$ we will refer to an \OR-cross-composition of cost~$f(t) = t^{1/d} \log (t)$ as a \emph{degree-$d$ cross-composition}. By Theorem~\ref{thm:cross_composition_LB}, a degree-$d$ cross-composition can be used to rule out generalized kernels of size~$\Oh(k^{d - \varepsilon})$. Note that when studying sparsification, we use the number of vertices or variables in the instance (which is usually denoted by~$n$) as the parameter value (which is usually denoted by~$k$).

When interpreting truth assignments as elements of a field, we equate the value \true with the $1$ element in the field (multiplicative identity), and the value \false with the $0$ element (additive identity). Consequently, for a boolean variable~$x$ its negation~$\neg x$ corresponds to~$(1-x)$. We let $\mathbb{Z}/m\mathbb{Z}$ denote the integers modulo~$m$, which form a field if~$m$ is a prime number. We use~$a \equiv_m b$ to denote that~$a$ and~$b$ are congruent modulo~$m$, and~$a \not\equiv_m b$ to denote non-congruence. We denote the greatest common divisor of a set~$S$ of non-negative integers by~$\gcd(S)$. The \emph{degree} of a multivariate polynomial is the maximum degree of its monomials. Let $f(x_1,\ldots, x_d)$ be a $d$-variate polynomial over a field $F$. The \emph{root set} of $f$ is the algebraic variety $\{(e_1,\ldots,e_d) \in F^d \mid f(e_1,\ldots,e_d) = 0\}$. For a field~$F$ and a finite set~$S \subseteq F$ of elements, the univariate polynomial~$f(x) := \prod_{s \in S} (x-s)$ over~$F$ of degree~$|S|$ has root set exactly~$S$. We say that a field~$F$ is \emph{efficient} if the field operations and Gaussian elimination can be done in polynomial time in the size of a reasonable input encoding. The field of rational numbers~$\mathbb{Q}$, and all finite fields, are efficient. We use~$[n]$ to denote~$\{1, \ldots, n\}$. We denote the positive and non-negative integers by~$\mathbb{N}$ and~$\mathbb{N}_0$, respectively. The~$\Ohtilde$-notation suppresses polylogarithmic factors: $\Ohtilde(n) = \Oh(n \log^{c} n)$ for a constant~$c$.

\section{Kernel upper bounds}
\subsection{Polynomial root CSP over a field}
We start by showing how to reduce the number of constraints in instances of \rootCSP, by extending the argument presented in the introduction.

\begin{theorem}\label{thm:kernel}
There is a polynomial-time algorithm that, given an instance~$(L,V)$ of \rootCSP over an efficient field~$F$, outputs an equivalent instance~$(L',V)$ with at most~$n^d+1$ constraints such that~$L' \subseteq L$.
\end{theorem}
\begin{proof}
Given a list~$L$ of polynomial equalities over variables~$V$ for \rootCSP, we use linear algebra to find redundant constraints. Observe that~$(x_i)^c = x_i$ for all~$0/1$-assignments and~$c \geq 1$. As constraints are evaluated over~$0/1$-assignments, we may assume without loss of generality that the monomials in each of the polynomials are multilinear: each monomial consists of a coefficient from~$F$ multiplied by distinct variables.

Create a matrix $A$ with $|L|$ rows and a column for every multilinear monomial of degree at most $d$ over variables from $V$. Let position $a_{i,j}$ in $A$ be the coefficient of the monomial corresponding to column $j$ in the polynomial equality corresponding to row $i$.

Compute a basis~$B$ of the row space of matrix $A$, for example using Gaussian elimination~\cite{Hogben13}, and let~$L'$ consist of the equalities in~$L$ whose corresponding row appears in the basis.
Since $L'\subseteq L$, it follows that if the original instance has a satisfying assignment, the reduced instance has a satisfying assignment as well. The crucial part of the correctness proof is to establish the converse.

\begin{claim}
If an assignment $\tau \colon V \rightarrow \{0,1\}$ of the variables in $V$ satisfies the equalities in $L'$, then it satisfies all equalities in~$L$.
\end{claim}
\begin{claimproof}
Consider any equality $(f(\vect{x}) = 0) \in L \setminus L'$, and assume it corresponds to the $i$'th matrix row. Let $f_j(\vect{x})$ be the polynomial represented in the $j$'th row of matrix $A$ for $j \in [|L|]$. Without loss of generality, let the basis of $A$ correspond to its first $m$ rows~$\vect{a}_1, \ldots, \vect{a}_m$. We then have $i>m$, and by the definition of basis there exist $\beta_1,\ldots,\beta_m \in F$ such that $\vect{a}_i = \sum_{j=1}^m \beta_j\vect{a}_j.$ Let $\vect{t}$ be the column vector containing, for each multilinear monomial of degree $\leq d$ in variables $x_1,\ldots,x_n$, the evaluation under $\tau$. For example, for monomial $x_1x_3$ it contains $\tau(x_1) \cdot \tau(x_3)$. By using the same order of monomials as in the construction of $A$, we obtain for all $j \in [|L|]$ that $f_j(\tau(x_1),\ldots,\tau(x_n)) = \vect{a}_j\vect{t}$, the inner product of $\vect{a}_j$ and $\vect{t}$. It follows that  $\vect{a}_j\vect{t} = 0$ for all $j \in [m]$, since satisfying $L'$ implies $f_j(\tau(x_1),\ldots,\tau(x_n)) = 0$. Now observe that
\[ f_i(\vect{x}) = \vect{a}_i\vect{t} = \sum_{j=1}^m (\beta_j\vect{a}_j)\vect{t} = \sum_{j=1}^m \beta_j(\vect{a}_j\vect{t}) = \sum_{j=1}^m \beta_j\cdot 0 = 0,\]
which proves the claim.
\end{claimproof}
\begin{claim}
The number of constraints in the resulting kernel is bounded by $n^d+1$.
\end{claim}
\begin{claimproof}
The size of a basis of any matrix over a field equals its rank, which is bounded by the number of columns. As there is a column for each multilinear monomial of degree at most $d$, there are at most~$\sum_{i=0}^d \binom{n}{i}$ constraints in the basis. Now observe that~$\sum_{i=1}^d \binom{n}{i} \leq n^d$. The left side counts nonempty subsets of~$[n]$ of size at most~$d$, each of which can be mapped to a distinct $d$-tuple by repeating an element. Since there are~$n^d$ $d$-tuples, the claim follows.
\end{claimproof}
This concludes the proof of Theorem~\ref{thm:kernel}.
\end{proof}

When each constraint can be encoded in~$\Ohtilde(n)$ bits, for example when each polynomial can be represented as an arithmetic circuit of size~$\Oh(n)$, Theorem~\ref{thm:kernel} gives a kernelization of size~$\Ohtilde(n^{d+1})$. When constraints can be encoded in~$\Ohtilde(1)$ bits, which may occur when constraints have constant arity, we obtain kernels of bitsize~$\Ohtilde(n^d)$. For explicit examples consider the following problem, where optionally a prime $p$ may be chosen.

\defparproblem{\problem{Generalized $d$-Sat (mod $p$)}}
{A set of clauses $\mathcal{C}$ over variables~$V := \{x_1, \ldots, x_n\}$, and for each clause a set $S_i \subset \mathbb{N}_0$ with $|S_i|\leq d$. Each clause is a set of distinct literals of the form~$x_i$ or~$\neg x_i$.}
{The number of variables $n$}
{Does there exist a truth assignment for the variables~$V$ such that the number of satisfied literals in clause $i$ modulo~$p$ lies in $S_i$ for all $i$?}

\begin{corollary}\label{cor:S_SAT}
\problem{Generalized $d$-Sat} and \problem{Generalized $d$-Sat mod $p$} both have a kernel with~$n^{d}+1$ clauses that can be encoded in~$\Oh(n^{d+1} \log n)$ bits.
\end{corollary}
\begin{proof}
To reduce the number of clauses using Theorem~\ref{thm:kernel}, we only have to provide a polynomial of degree at most~$d$ to represent each constraint. Consider a clause involving~$k$ variables~$x_{i_1}, \ldots, x_{i_k}$, with set $S_\ell$. Let~$t_j = x_{i_j}$ if variable~$x_{i_j}$ occurs positively in the clause, and let~$t_j = (1-x_{i_j})$ if the variable occurs negatively. Then the number of satisfied literals in the clause is given by the degree-1 polynomial
$f(x_{i_1},\ldots,x_{i_k}) := \sum_{i=1}^k t_i.$
Let $F(x)$ be a polynomial with root set $S_\ell$ (mod $p$) of degree at most $|S_\ell|$. We obtain~$F(f(\vect{x})) \equiv_p 0$ if and only if~$\vect{x}$ satisfies the clause. Note that the degree of $F(f(\vect{x}))$ is at most $|S_\ell| \leq d$.

Applying Theorem~\ref{thm:kernel} to the resulting instance of \rootCSP identifies a subset of at most~$n^d+1$ constraints which preserve the answer to the \textsc{Sat} problem. Each clause contains at most~$2n$ literals, which can be encoded in~$\Oh(\log n)$ bits each. Additionally, for each clause we need to store the set~$S_\ell$ of at most~$d$ integers, which have value at most~$2n$ in relevant inputs. As~$d$ is a constant, the instance can be encoded in~$\Oh(n^{d+1} \log n)$ bits.
\end{proof}

Corollary~\ref{cor:S_SAT} yields a new way to get a nontrivial compression for \dNAESAT, which is conceptually simpler than the existing approach which requires an unintuitive lemma by Lov\'asz~\cite{Lovasz76}. The new approach gives the same size bound as given earlier~\cite{sparsificationJansenP15}.

\begin{corollary}
\dNAESAT has a kernel with $n^{d-1}+1$ clauses that can be encoded in~$\Oh(n^{d-1}\log{n})$ bits.
\end{corollary}
\begin{proof}
A clause of size $k\leq d$ is not-all-equal satisfied if and only if the number of satisfied literals lies in $S:=\{1,\ldots,k-1\}$. Using Corollary \ref{cor:S_SAT} we can reduce the number of clauses to $n^{d-1}+1$. Each clause has $d \in O(1)$ variables and can thus be encoded in $\Oh(\log n)$ bits.
\end{proof}

\begin{note}

\subsection{Polynomial root CSP modulo a non-prime}
We can generalize Theorem \ref{thm:kernel} to also obtain a sparsification for \rootCSP over the integers modulo a non-prime. We give two different approaches for sparsifying such problems. The first approach gives the smallest number of constraints after reduction, but has the disadvantage that the resulting list of constraints is not necessarily a subset of the original list of constraints. The second approach results in a larger (but still bounded) number of constraints, which form a subset of the original constraints. We first give some linear-algebraic background.

Consider an instance~$(L,V)$ of \rootCSP over a ring~$R$ with~$n$ variables and~$m$ constraints. We consider the matrix~$A$ over~$R$ with~$m$ rows and~$\sum_{i=0}^d \binom{n}{d}$ columns, in which the $i$th row contains the coefficients of the multilinear monomials in the polynomial for the $i$th constraint. The satisfiability of the constraints by a $0/1$-assignment then comes down to the following question: is there a $0/1$-assignment to the variables, such that the vector~$\vect{x}$ consisting of all multilinear monomial evaluations of the variables~$x_1, \ldots, x_n$ satisfies~$A\vect{x} = \vect{0}$ over~$R$? The key insight for the sparsification is that any matrix~$B$ for which the row-space over~$R$ is equal to that of~$A$, satisfies~$A\vect{x} = \vect{0} \Leftrightarrow B\vect{x} = \vect{0}$. (Recall that the row-space over~$R$ consists of the vectors that can be written as a linear combination of the rows, with coefficients from~$R$.) Hence we can obtain an encoding of an equivalent problem by selecting a matrix~$B$ whose row-space equals that of~$A$. When working over a field we can just extract a basis for the row-space to obtain~$B$, which is exactly what happened in Theorem~\ref{thm:kernel}. When working over the integers modulo~$m$ for composite~$m$, the existence of a basis is not guaranteed. For our first approach we therefore use the Howell normal form of the matrix, which is a canonical matrix form which has the same row-space.

\begin{theorem}\label{thm:kernel_modm}
There is a polynomial-time algorithm that, given an instance~$(L,V)$ of \rootCSP over $\mathbb{Z}/m\mathbb{Z}$ for some integer $m \geq 2$, outputs an equivalent instance~$(L',V)$ of \rootCSP over $\mathbb{Z}/m\mathbb{Z}$ with at most $n^d+1$ constraints.
\end{theorem}
\begin{proof}
In a similar way as in Theorem \ref{thm:kernel}, we use linear algebra to find redundant constraints. Let a list $L$ of polynomial equalities over variable set $V$ be given. We again assume that the monomials in each of the polynomials are multilinear. Construct a matrix $A$ with $|L|$ rows and a column for every multilinear monomial of degree at most $d$ over variables from $V$. Let position $a_{i,j}$ in $A$ contain the coefficient of the monomial
corresponding to column $j$ in the polynomial equality corresponding to row $i$. 

We now compute the Howell form $H$ of matrix $A$, which was first defined by Howell \cite{Howell86}, such that $A = PH$, where $P$ is invertible over $\mathbb{Z}/m\mathbb{Z}$. This can be done in polynomial time, see for example \cite[\S 3]{StorjohannM98}.
Let $H'$ be the matrix $H$ with all zero rows removed and let $L'$ contain the polynomial equations given by the rows of $H'$. We now prove the correctness of this procedure.

\begin{claim}
An assignment $\tau:V\rightarrow \{0,1\}$ of the variables in $V$ satisfies the equalities in $L'$, if and only if it satisfies the equalities in $L$.
\end{claim}
\begin{claimproof}
$(\Rightarrow)$ Suppose assignment $\tau:V\rightarrow \{0,1\}$ satisfies all equalities in $L'$. Consider the vector $\vect{x}$ with the assignment given to the $j$'th monomial on position $j$. Then
\[H'\vect{x} = \vect{0} \Leftrightarrow H\vect{x} = \vect{0}
\Rightarrow  PH\vect{x} = P\vect{0}\Rightarrow A \vect{x} = \vect{0},\]
which implies that $\tau$ is also a satisfying assignment for $L'$.

$(\Leftarrow)$  Suppose assignment $\tau:V\rightarrow \{0,1\}$ satisfies all equalities in $L$. Consider the vector $\vect{x}$ with the assignment given to the $j$'th monomial on position $j$. Then
\begin{align*}
A\vect{x} = \vect{0} \Rightarrow
PH\vect{x} = \vect{0} \Rightarrow
P^{-1}PH\vect{x} = P^{-1}\vect{0} \Rightarrow H\vect{x} = \vect{0} \Rightarrow H'\vect{x} = \vect{0},
\end{align*}
which implies that $\tau$ is also a satisfying assignment for $L'$.
\end{claimproof}

\begin{claim}
The number of constraints in the resulting kernel $L'$ is bounded by $n^d+1$.
\end{claim}
\begin{claimproof}
The number of constraints in $L'$ equals the number of rows in $H'$. We will use the following properties of a matrix in Howell form \cite[\S 3]{StorjohannM98} to give an upper bound on the number of non-zero rows in $H$.
\begin{itemize}
\item Let $r$ be the number of non-zero rows of $H$. Then the first $r$ rows of $H$ are non-zero.
\item 
For $1 \leq i \leq r$ let the first non-zero entry in row $i$ of $H$ be in column $j_i$. Then $j_1 < j_2 < \ldots < j_r$.
\end{itemize}
By these two properties any matrix in Howell form has at most as many non-zero rows as it has columns. Thereby there are at most $n^d + 1$ polynomial equations in $L'$.
\end{claimproof}
This concludes the proof of Theorem \ref{thm:kernel_modm}.
\end{proof}

Compared to Theorem~\ref{thm:kernel}, the sparsification of Theorem~\ref{thm:kernel_modm} has the disadvantage that it may output polynomials (representing constraints) that were not part of the input. If the input polynomials had an efficient encoding, for example as an arithmetic circuit, this property may be lost in the transformation. In general, to represent an output polynomial one may have to store all its~$\Oh(n^d)$ coefficients individually. We present an alternative approach that alleviates this issue by ensuring that the set of constraints in the output instance is a subset of the original constraints. However, it comes at the expense of increasing the number of constraints. The following lemma captures the key linear-algebraic insight behind the approach.

\begin{lemma} \label{lemma:spanningset:modm}
Let~$m \geq 2$ be an integer with~$r$ distinct prime divisors. For any~$S \subseteq \mathbb{Z}/m\mathbb{Z}$ there exists a subset~$S' \subseteq \mathbb{Z}/m\mathbb{Z}$ of size at most~$r$ such that any element in~$S$ can be written as a linear combination over~$\mathbb{Z}/m\mathbb{Z}$ of elements in~$S'$. For any fixed~$m$, one can compute~$S'$ and expressions of all~$a \in S$ as linear combinations of~$S'$ in polynomial time.
\end{lemma}
\begin{proof}
Let~$p_1, \ldots, p_r$ be the distinct prime divisors of~$m$, which can be found in constant time for fixed~$m$. For a prime~$p$ and positive integer~$a$, define:
\begin{align*}
\mu_p(a) &:= \max \{ k \in \mathbb{N}_0 \mid p^k \mbox{ divides $a$}\}. \\
\nu_p(a) &:= \max \{ k \in \mathbb{N}_0 \mid p^k \mbox{ divides both $a$ and $m$}\}.
\end{align*}
Observe that~$\nu_p(a) \leq \mu_p(a)$ for all~$a$. For any~$a$ that divides~$m$ we have~$\nu_p(a) = \mu_p(a)$.

Using these notions we construct the set~$S'$ as follows. For each~$i \in [r]$ select an element~$a \in S$ that minimizes~$\nu_{p_i}(a)$ and add this element to~$S'$. Since~$m$ is constant this can be done in polynomial time. The resulting set~$S'$ has size at most~$r$. We prove it spans~$S$ using the following claim.

\begin{claim} \label{claim:gcd:divides:set}
Let~$d$ be the largest integer that simultaneous divides~$m$ and all elements of~$S'$. For any~$b \in S \setminus S'$, the integer~$d$ divides~$b$. Equivalently:~$\gcd(m,S') \mid b$.
\end{claim}
\begin{claimproof}
If~$d=1$ then the claim is trivial. Suppose all prime factors~$p$ of~$d$ are also prime factors of~$b$ with~$\mu_p(d) \leq \mu_p(b)$. Then the factorization of~$b$ can be written as the factorization of~$d$ multiplied by remaining factors. Hence~$d \mid b$, and the claim follows.

Now suppose there is a prime factor~$p$ of~$d$ with~$\mu_p(d) > \mu_p(b)$. Since~$p$ is a factor of~$d = \gcd(m,S')$, we know~$p$ is a factor of~$m$ and was therefore considered during the construction of~$S'$. Since~$d$ divides~$m$ we know that~$\mu_p(d) = \nu_p(d)$. Combined with the fact that~$\nu_p(b) \leq \mu_p(b)$ it follows that~$\nu_p(b) \leq \mu_p(b) < \mu_p(d) = \nu_p(d)$. Since~$d$ divides all members of~$S'$, it follows that~$\nu_p(b) < \nu_p(d) \leq \nu_p(a)$ for all~$a \in S'$. But then~$b$ should have been added to~$S'$ during its construction; a contradiction.
\end{claimproof}

To conclude the proof, we use Claim~\ref{claim:gcd:divides:set} to show that any~$b \in S \setminus S'$ can efficiently be written as a linear combination of~$S'$ over~$\mathbb{Z}/m\mathbb{Z}$. By B\'{e}zout's identity, the greatest common divisor of a set of integers can be written as an integer linear combination of the elements in that set. Such a combination can efficiently be found using the extended Euclidean algorithm. Hence there are integer coefficients~$\alpha_{\cdot}$ such that $d = \gcd(m,S') = \alpha_m \cdot m  + \sum _{a \in S'} \alpha_a \cdot a$. Let~$b' := b \div \gcd(m,S')$, which is integral by Claim~\ref{claim:gcd:divides:set}. But then
\[b = b' \cdot \gcd(m,S') = (b' \cdot \alpha_m) m + \sum _{a \in S'} (b' \cdot \alpha_a) a,\]
which implies that
\[b \equiv_m \sum _{a \in S'} (b' \cdot \alpha_a) a \equiv_m \sum _{a \in S'} ((b' \cdot \alpha_a) \bmod m) a\]
is a linear combination over~$\mathbb{Z}/m\mathbb{Z}$ resulting in~$b$.
\end{proof}

The following lemma follows from a procedure similar to Gaussian elimination, using Lemma~\ref{lemma:spanningset:modm} as a subroutine.

\begin{lemma} \label{lemma:basis:modm}
Let~$m \geq 2$ be an integer with~$r$ distinct prime divisors. For any matrix~$A$ over $\mathbb{Z}/m\mathbb{Z}$ in which~$k \geq 1$ columns contain a nonzero element, there is a subset~$B$ of~$r \cdot k$ rows of~$A$ that spans the row-space of~$A$. For any fixed~$m$, such a subset~$B$ can be found in polynomial time.
\end{lemma}
\begin{proof}
Proof by induction on~$k$. Consider the first column~$c_i$ of~$A$ that contains a nonzero and let~$S$ be the elements appearing in that column. Using Lemma~\ref{lemma:spanningset:modm}, compute a subset~$S' \subseteq S$ of size at most~$r$ that spans~$S$, and find the corresponding linear combinations. For each element~$a \in S'$ select one row with value~$a$ in column~$c_i$ and add it to~$B_1$. If~$c_i$ is the only column containing a nonzero, then it is easy to see that~$B := B_1$ is a valid output for the procedure. Otherwise, since all elements of~$S$ are linear combinations of elements of~$S'$, by subtracting the relevant linear combinations of rows of~$B_1$ from rows in~$A$ we can obtain zeros at all positions in column~$c_i$, without introducing nonzeros in earlier columns. Let~$A'$ be the resulting matrix, which therefore has at most~$k-1$ nonzero columns. Apply induction to find a spanning subset~$B'$ of the rows of~$A'$ of size at most~$r \cdot (k-1)$. Let~$B_2$ be the rows of~$A$ corresponding to rows~$B'$ in~$A'$. Then~$B_1 \cup B_2$ consists of at most~$r + (k-1)r = r k$ rows of~$A$. It is easy to verify that these rows indeed span the row-space of~$A$. The inductive proof directly translates into a polynomial-time recursive algorithm, using the fact that the procedure of Lemma~\ref{lemma:spanningset:modm} provides the required linear combinations.
\end{proof}

\begin{theorem}\label{thm:subset_kernel_modm}
There is a polynomial-time algorithm that, given an instance~$(L,V)$ of \rootCSP over $\mathbb{Z}/m\mathbb{Z}$ for some fixed integer $m \geq 2$ with~$r$ distinct prime divisors, outputs an equivalent instance~$(L',V)$ of \rootCSP over $\mathbb{Z}/m\mathbb{Z}$ with at most $r\cdot(n^d+1)$ constraints such that~$L' \subseteq L$.
\end{theorem}
\begin{proof}
We proceed similarly as in the proof of Theorem~\ref{thm:kernel_modm}. Consider an input~$(L,V)$ of \rootCSP over $\mathbb{Z}/m\mathbb{Z}$ with~$n := |V|$ variables. Let~$A$ be the matrix with~$|L|$ rows and~$\sum _{i=0}^d \binom{n}{i}$ columns, containing the coefficients of the multilinear monomials that form the constraints for each of the~$|L|$ constraint polynomials. A $0/1$-assignment to the variables satisfies all constraints if and only if the vector~$\vect{x}$ of all monomial evaluations satisfies~$A\vect{x} = \vect{0}$. Use Lemma~\ref{lemma:basis:modm} to compute a subset~$B$ of at most~$r \cdot \sum _{i=0}^d \binom{n}{i} \leq r \cdot (n^d + 1)$ rows of~$A$ that span the row-space of~$A$. Let~$L'$ contain the constraints whose corresponding row appears in~$B$ and output the instance~$(L',V)$ as the result of the procedure. Using the guarantee of Lemma~\ref{lemma:basis:modm} this procedure runs in polynomial time for fixed~$m$. Since~$L' \subseteq L$, the instance~$(L',V)$ can be satisfied if~$(L,V)$ can. For the reverse direction, consider a satisfying assignment for~$(L',V)$ and the corresponding vector~$\vect{x}$ of evaluations of multilinear monomials of degree at most~$d$. Then~$B\vect{x} = \vect{0}$ since the assignment satisfies all constraints in~$L'$. As any row in~$A$ can be written as a linear combination of rows in~$B$, it follows that~$A\vect{x} = \vect{0}$, showing that~$(L,V)$ is satisfiable and hence that the output instance is equivalent to the input.
\end{proof}

\end{note}
\subsection{Polynomial non-root CSP} \label{sec:UB:non-root}
In this section we consider \nonrootCSP. In Section~\ref{sec:nonroot:lb} we will show that, over the field of rational numbers, the problem cannot be compressed to size polynomial in~$n$, unless \containment. We therefore consider the field \primefield of integers modulo a prime~$p$.

\begin{theorem}\label{thm:UB:non-root}
There is a polynomial-time algorithm that, given an instance~$(L,V)$ of \nonrootCSP over \primefield, outputs an equivalent instance~$(L',V)$ with at most~$n^{d(p-1)}+1$ constraints such that~$L' \subseteq L$.
\end{theorem}
\begin{proof}
Suppose we are given a list of polynomial inequalities $L$ over variables $V$. Observe that an inequality $f(\vect{x}) \not \equiv_p 0$ is equivalent to
$f(\vect{x}) \bmod p \in \{1,\ldots,p-1\}$.

Let $F \colon \primefield \to \primefield$ be a polynomial of degree $p-1$ with root set $\{1,\ldots,p-1\}$ modulo $p$, which exists since~$\primefield$ is a field. Then~$f(\vect{x}) \not \equiv_p 0$ can equivalently be stated as~$F(f(\vect{x})) \equiv_p 0$. It is easy to see that $F(f(\vect{x}))$ is a polynomial of degree at most $d(p-1)$. Therefore, $L$ can be written as an instance of \problem{$d(p-1)$-Polynomial root CSP} by replacing every polynomial $f$ by $F\circ f$. By Theorem~\ref{thm:kernel}, the theorem follows.
\end{proof}

\noindent In Section~\ref{sec:nonroot:lb} we will establish a nearly-matching lower-bound counterpart to Theorem~\ref{thm:UB:non-root}. We do not have upper bounds for \nonrootCSP modulo a composite number~$m$. The difficulties in obtaining these are described in Section~\ref{sect:conclusion}.

\section{Kernel lower bounds}
\subsection{Polynomial root CSP over the rationals}
We now turn our attention to lower bounds, starting with \rootCSP over~$\mathbb{Q}$ and over $\mathbb{Z}/m\mathbb{Z}$. We start by proving that \problem{Exact Red-Blue Dominating Set} does not have generalized kernels of bitsize~$\Oh(n^{2-\varepsilon})$ for any~$\varepsilon > 0$, unless \containment. The same lower bound for both variants of \linearrootCSP will follow by a linear-parameter transformation. We then show how to generalize this result to \rootCSP. As a starting problem for the cross-composition we will use the NP-hard \problem{Red-Blue Dominating Set (\RBDS)}~\cite{DomLS14,Karp72}.

\defproblem{Red-Blue Dominating Set (\RBDS)}
{A bipartite graph $G = (R \cup B, E)$ containing red ($R$) and blue ($B$) vertices, and an integer~$k$.}
{Does there exist a set $D \subseteq R$ with $|D| \leq k$ such that every vertex in $B$ has at least one neighbor in $D$?}

\problem{Exact Red Blue Dominating Set (\ERBDS)} is defined similarly, except that every vertex in $B$ must have \emph{exactly one} neighbor in $D$. Furthermore we will not bound the size of such a set, but merely ask for the existence of any \ERBDS. 
Finally, we define a weakening of the notion of an \ERBDS of a graph, called a \SERBDS. Given a bipartite graph $G$ and set $S \subseteq V(G)$, a set $X \subseteq V(G)$ is a \SERBDS of $G$ with respect to $S$ if it is a \RBDS of $G$ and furthermore, any blue vertex $x \notin S$ has exactly one neighbor in $X$. Vertices from $S$ may be dominated multiple times.

The following lemma gives a degree-$2$ cross-composition from \RBDS to \SERBDS, which will be used to prove Theorem \ref{thm:exact_RBDS}. It is proven separately because the construction will also be used in the proofs of Theorems \ref{cor:LB:1-poly-root-csp-modm} and \ref{thm:drootcspmod:lb}.

\begin{lemma}\label{lem:cross_composition_erbds}
There exists a polynomial-time algorithm that, given $t$ instances of \RBDS with $\sqrt{t} \in \mathbb{N}$, labeled $X_{\ell_1,\ell_2}$ with $\ell_1,\ell_2 \in [\sqrt{t}]$, which all ask for a solution of size $k$ and all have $m_R$ red and $m_B$ blue vertices, constructs a bipartite graph $G'$ with vertices partitioned into red $(R)$ and blue $(B)$ vertices, and a subset $V$ of the blue vertices, such that the following holds:
\begin{enumerate}
\item \label{req:size} $|R| + |B| \leq \Oh (\sqrt{t} \cdot (m_R + m_B)^3)$.
\item \label{req:RBDS_ERBDS} If there exist $\ell_1,\ell_2\in[\sqrt{t}]$ such that $X_{\ell_1,\ell_2}$ has a \RBDS of size $k$, then $G'$ has an \ERBDS.
\item \label{req:SERBDS_RBDS} If $G'$ has a \SERBDS with respect to $V$, then there exist $\ell_1,\ell_2 \in[\sqrt{t}]$ such that $X_{\ell_1,\ell_2}$ has a \RBDS of size $k$.
\item \label{req:degrees} There are at most $2$ vertices in $B\setminus V$ with degree more than $m_R + k + 2$.
\end{enumerate}
\end{lemma}
In particular, the lemma shows how to embed a series of $t$ size-$n$
instances $X_{\ell_1, \ell_2} = (G_{\ell_1, \ell_2}, k)$ for $\ell_1, \ell_2 \in [\sqrt t]$ that share the same target value $k$,  into a single graph $G'$ with $\Oh(\sqrt{t} \cdot \text{poly}(n))$ vertices such that $G'$ has an \ERBDS if and only if some input instance has a size-$k$ \RBDS. This straightforwardly gives a sparsification lower bound for \ERBDS: since the number of output vertices is roughly $\sqrt{t}$, by choosing a suitable polynomial equivalence relation we get a degree-$2$ cross composition. Now the actual lemma statement is even stronger than the statement ``some input has a \RBDS $\Leftrightarrow$ $G'$ has an \ERBDS'', because the $(\Leftarrow)$ implication already holds when $G'$ has a \SERBDS. The fact that it is only required to be exact on a set of vertices $B\setminus V$ that has almost only small-degree vertices, will be used later. Later constructions ``pay extra'' for checking exactness of large-degree vertices, and the bound in (\ref{req:degrees}) guarantees this does not happen too often.

Before proving the lemma, let us give the main ideas.
The standard approach to give a degree-$2$ cross composition~\cite{DellM12Kernelization,DomLS09Incompressibility,sparsificationJansenP15} is to have a table-like structure with sets of vertices $U_\ell$ consisting of~$m_R$ vertices and~$V_\ell$ consisting of~$m_B$ vertices for all $\ell \in [\sqrt{t}]$. In this way we can add connections between $U$ and $V$ such that $G'[U_{\ell_1} \cup V_{\ell_2}]$ is isomorphic to $G_{\ell_1,\ell_2}$, thereby embedding the adjacency information of all $t$ individual inputs while only needing $\sqrt{t} \cdot (m_R + m_B)$ vertices in the graph. Selector gadgets are then used to ensure that the part of a \textsc{(semi)-erbds} in $G'$ in $U_{\ell_1}$ for some $\ell_1$ corresponds to a \RBDS of size $k$ in $G_{\ell_1,\ell_2}$ for some $\ell_2$. In our case however, difficulties arise when we try to use this type of construction. Given a \RBDS for some input instance $G_{\ell_1,\ell_2}$, finding an \ERBDS in $G'$ can be problematic. The issue is that adding the vertices in $U_{\ell_1}$ corresponding to a solution in $G_{\ell_1,\ell_2}$ to an \ERBDS in $G'$, may dominate some of the vertices from $V$ multiple times. This is not easy to avoid, as there is simply no guarantee on how many times a vertex in the set $V_{\ell}$ with $\ell \neq \ell_2$ will be dominated by this choice of red vertices.

To resolve this problem, every set $U_{\ell}$ and $V_{\ell}$ has $k$ copies of each vertex. Connections are made such that the $i$'th copy of a vertex may only connect to the $i$'th copy of another vertex, such that $G[U_{\ell_1} \cup V_{\ell_2}]$ contains $k$ disjoint copies of $G_{\ell_1,\ell_2}$. To translate a \RBDS in $G_{\ell_1,\ell_2}$ to a \ERBDS in $G'$, we take at most one vertex from the $i$'th set of copies in $U_{\ell_1}$. Hereby, any vertex in $V$ is dominated at most once.
Furthermore, for each vertex in $V_{\ell_2}$, at least one of its copies is dominated. We add additional gadgets to ensure that the remaining vertices can also be dominated.

\begin{proof}[{Proof of Lemma \ref{lem:cross_composition_erbds}.}]
Let instance $X_{\ell_1,\ell_2}$ have graph $G_{\ell_1,\ell_2}$, with red vertices $R_{\ell_1,\ell_2}$ and blue vertices $B_{\ell_1,\ell_2}$. For each input graph $G_{\ell_1,\ell_2}$ enumerate the red vertices as $r_1,\ldots,r_{m_R}$ and the blue vertices as $b_1,\ldots,b_{m_B}$, arbitrarily. Create a graph $G'$ by the following steps. Figure \ref{fig:overview} shows a sketch of $G'$.

\begin{figure}
\begin{center}
\includegraphics{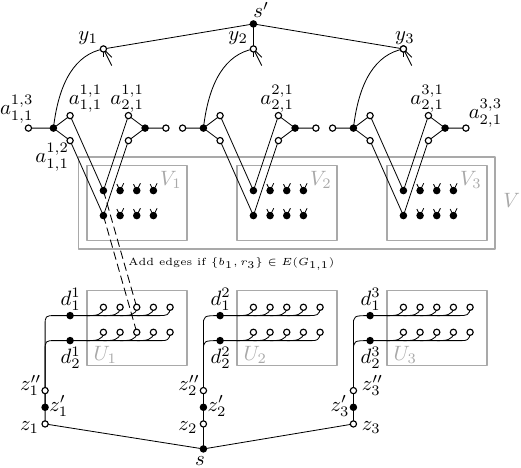}
\caption{The graph $G'$ created in the proof of Lemma \ref{lem:cross_composition_erbds}, for $k=2$, $m_R=5$, $m_B=4$, and $t=9$. Edges between $U$ and $V$ are left out for simplicity. Of the $24$ gadgets in $C$ only $c^\ell_{1,1}$ and $c^\ell_{2,1}$ are shown for all $\ell \in [\sqrt{t}]$. Vertices in $R$ are shown in white and vertices in $B$ are shown in black. The set $V$ of vertices that may be dominated multiple times by a \SERBDS wrt. $V$ is highlighted by a rectangle. }
\label{fig:overview}
\end{center}
\end{figure}

\begin{enumerate}
\item \label{step:create_red} Create $\sqrt{t}$ sets $U_1,\ldots,U_{\sqrt{t}}$ each consisting of $k \cdot m_R$ red vertices, with $U_\ell~:= \{u^\ell_{i,j}~\mid~i~\in~[k],$
    $j~\in~[m_R]\}$ for each~$\ell \in [\sqrt{t}]$. Let $U$ be the union of all sets $U_\ell$, for $\ell \in[\sqrt{t}]$.
\item \label{step:create_blue} Similarly create $\sqrt{t}$ sets $V_1,\ldots,V_{\sqrt{t}}$, each consisting of $k \cdot m_B$ blue vertices, and define $V_\ell~:=~\{v^{\ell}_{i,j'}\mid i\in[k],j'\in[m_B]\}$ for all $\ell \in [\sqrt{t}]$. Let $V$ be the union of all sets $V_\ell$. Note that a \SERBDS wrt. $V$ must dominate all blue vertices that are created in the remainder of the construction exactly once.
\item \label{step:connect_red_blue} For each $i \in [k]$ add the edge from $u^{\ell_1}_{i,j}$ to $v^{\ell_2}_{i,j'}$ if  $\{r_j,b_{j'}\}$ is an edge in instance $X_{\ell_1,\ell_2}$ with $\ell_1,\ell_2 \in [\sqrt{t}]$, $j\in[m_R]$, and $j'\in[m_B]$.
\end{enumerate}
    By Steps \ref{step:create_red} to \ref{step:connect_red_blue}, the subgraph of $G'$ induced by the vertices in $U_{\ell_1} \cup V_{\ell_2}$ consists of $k$ vertex-disjoint copies of $G_{\ell_1,\ell_2}$. The next steps are used to ensure that there are exactly $k$ vertices from $U$ in any \SERBDS, which must all belong to the same set $U_\ell$. These vertices will correspond to a \RBDS in one of the input instances.
\begin{enumerate}
\setcounter{enumi}{3}
\item\label{step:d} Create blue vertices  $d^\ell_i$ for $\ell \in [\sqrt{t}]$ and $i \in [k]$. Connect vertex $d^\ell_i$ to all vertices $u^\ell_{i,j}$ with $j\in[m_R]$. Define~$D:=\{d^\ell_i\mid \ell \in [\sqrt{t}], i \in [k]\}$. These blue vertices ensure that a \SERBDS wrt. $V$, which dominates each vertex of $D$ exactly once, cannot contain two vertices~$u^\ell_{i,j}$ and~$u^\ell_{i,j'}$ belonging to the same row of the same set~$U_\ell$.
\item Add blue vertex $s$ and  add the vertices $Z := \{z_\ell,z_\ell',z_\ell'' \mid \ell \in [\sqrt{t}]\}$. Let $z_\ell$ and $z''_\ell$ be red and let $z'_\ell$ be blue for all $\ell \in [\sqrt{t}]$. Connect $z_\ell''$ to $d^{\ell}_i$ for $i \in [k]$ and $\ell \in [\sqrt{t}]$. Add the edges $\{z_\ell,z_\ell'\}$ and $\{z_\ell',z_\ell''\}$ for all $\ell \in [\sqrt{t}]$.
     Connect each vertex $z_\ell$ to $s$ for $\ell \in [\sqrt{t}]$, thereby ensuring that exactly one vertex $z_\ell$ is contained in a \SERBDS wrt.~$V$.
    Intuitively, the index~$\ell_1$ for which~$z_{\ell_1}$ belongs to a \SERBDS controls the first index of the input instance~$X_{\ell_1, \ell_2}$ to which the solution corresponds. \label{step:makeS} 
\end{enumerate}
The next steps ensure that some of the blue vertices in one set $V_{\ell_2}$ need to be dominated by vertices from $U$, while all other vertices in $V$ can be dominated ``for free''.
This will control the second index of the input instance $X_{\ell_1, \ell_2}$ to which the solution corresponds.
\begin{enumerate}
\setcounter{enumi}{6}
\item \label{step:select_blue} Add sets of gadgets $C_\ell$ for $\ell \in [\sqrt{t}]$. Each set~$C_\ell$ consists of $m_B \cdot k$ selector gadgets $c^\ell_{i,j'}$ for $i\in[k]$, $j'\in[m_B]$. Selector gadget $c^\ell_{i,j'}$ consists of $k+1$ red vertices labeled $\A{\ell}{1}{i}{j'},\ldots,\A{\ell}{k+1}{i}{j'}$ that are all connected to a blue vertex $\B{\ell}{i}{j'}$ that is the only blue vertex inside the gadget. Furthermore, for $j' \in [m_B]$, $\ell \in [\sqrt{t}]$ and $i\in[k]$, in gadget $c^\ell_{i,j'}$ the vertex $\A{\ell}{x}{i}{j'}$ for $x \in [k]$ is connected to $v^\ell_{x,j'}$. We refer to the vertex set of gadget $c^\ell_{i,j'}$ by $V(c^\ell_{i,j'})$.
\end{enumerate}
By Step \ref{step:select_blue} of the construction a \SERBDS uses at most one red vertex from each gadget, which can be used to dominate one vertex from $V$. Using vertex $\A{\ell}{k+1}{i}{j'}$ of a gadget, the blue vertex of that gadget can be dominated without dominating any other blue vertices. Using the $k$ gadgets introduced for $j' \in [m_B],\ell\in[\sqrt{t}]$, we can thus precisely dominate all vertices in $\{v_{i,j'}^\ell \mid i \in [k]\}$. Now we will ensure that there is a $\ell_2 \in [\sqrt{t}]$ such that in $V_{\ell_2}$, for each $j' \in [m_B]$, one of the vertices $\{v_{i,j'}^{\ell_2} \mid i \in [k]\}$ is not dominated by a gadget and must therefore be dominated by a vertex from $U$.
\begin{enumerate}
\setcounter{enumi}{7}
\item Add red vertices $Y:=\{y_1,\ldots,y_{\sqrt{t}}\}$. For each $\ell\in[\sqrt{t}]$ connect $y_\ell$ to the blue vertices of the gadgets $c^\ell_{1,j'}$ for all $j'\in [m_B]$, thereby making gadget $c^\ell_{i,j'}$ special for $i=1$. Connect $y_1,\ldots,y_{\sqrt{t}}$ to the new blue vertex $s'$, which ensures that exactly one vertex~$y_{\ell_2} \in Y$ belongs to any \SERBDS wrt.~$V$. \label{step:makeSprime}
\end{enumerate}
%
This concludes the construction of graph $G'$, with red vertices $R := (U \cup Y \cup \{z_\ell,z_{\ell}'' \mid \ell \in [\sqrt{t}]\}\cup \{\A{\ell}{x}{i}{j'} \mid {\ell\in[\sqrt{t}]},\ {x\in[k+1]},\ {i\in[k]},\ {j'\in[m_B]}\})$, and blue vertices $B := (V \cup D \cup \{s,s'\} \cup \{\B{\ell}{i}{j'} \mid {\ell\in[\sqrt{t}]},{i\in[k]},\ j'\in[m_B]\}\cup \{z_\ell ' \mid \ell \in [\sqrt{t}]\})$. The following observation follows immediately from the construction above.

\begin{observation}\label{obs:same_nbhood}
Let $x \notin V$ be a blue vertex in $G'$. The neighborhood of $x$ depends only on $m_{R}$, $m_{B}$, $t$, and $k$; it is independent of the structure of the given input instances.
\end{observation}

Furthermore, we can show that  requirement \ref{req:degrees} of this lemma is satisfied.
\begin{claim}\label{obs:degree_of_nbhood}
There are at most $2$ vertices in $B\setminus V$ with degree more than~$m_R + k + 2$.
\end{claim}
\begin{claimproof}
We list all vertices in $B \setminus V$, together with an upper bound on their degree.
\begin{description}
\item[{\em Vertices $s$ and $s'$}:] It follows from Steps \ref{step:makeS} and \ref{step:makeSprime} that these two vertices both have large degree, namely~$\sqrt{t}$.
\item[{\em Vertices in $D$}:]  It follows from Steps \ref{step:d} and \ref{step:makeS} that these vertices have degree $m_R + 1$.
\item[{\em Vertices in $Z \cap B$:}] It follows from Step \ref{step:makeS} that vertex $z_\ell'$ has degree two for all $\ell \in [\sqrt{t}]$.
\item[{\em Vertices in gadgets}:]   The blue vertex of any gadget has degree at most $k+2$, the incident edges are added in Steps \ref{step:select_blue} and \ref{step:makeSprime}.
\end{description}
Thus there are at most $2$ vertices of degree larger than $m_R + k + 2$ in $B \setminus V$.
\end{claimproof}

\begin{claim}\label{claim:all_from_one_square}
For any \SERBDS $E$ of $G'$ wrt.~$V$, there exists an index $\ell_1 \in [\sqrt{t}]$ such that $U_{\ell} \cap E = \emptyset$ for all $\ell \neq \ell_1 \in [\sqrt{t}]$ and $|E \cap \{u^{\ell_1}_{i,j}\mid  j \in [m_R]\}| = 1$ for all $i \in [k]$.
\end{claim}
\begin{claimproof}
By Step~\ref{step:makeS}, blue vertex~$s \notin V$ has neighborhood $\{z_\ell \mid\ \ell\in[\sqrt{t}]\}$. Since the semi-exact RBDS is exact on blue vertices outside~$V$, exactly one neighbor of~$s$ is contained in $E$; let this be~$z_{\ell_1}$. Thereby, for all $\ell \in [\sqrt{t}]$ with $\ell \neq \ell_1$ we obtain $z_\ell \notin E$. Since blue vertex $z_\ell'$ has neighborhood exactly $N_{G'}(z_\ell') = \{z_\ell,z_\ell ''\}$, it follows that $z_\ell'' \in E$ for all $\ell \neq \ell_1$ with $\ell \in [\sqrt{t}]$.

Let $\ell \in [\sqrt{t}]$ with $\ell \neq \ell_1$, we show that no vertex in $U_\ell$ is in $E$. Consider vertex $u_{i,j}^\ell$ with $i \in [k]$, $j\in [m_R]$. Then $u_{i,j}^\ell \in N_{G'}(d_{i}^\ell)$ for blue vertex $d_i^\ell$. Since $z_\ell'' \in N_{G'}(d_i^\ell)$ and $z_\ell'' \in E$, it follows that $u_{i,j}^\ell \notin E$.

It remains to show that $|E \cap \{u_{i,j}^{\ell_1} \mid j \in [m_r]\}| = 1$ for all $i \in [k]$. Since $N_{G'}(z_{\ell_1}') = \{z_{\ell_1},z_{\ell_1} ''\}$ and $z_{\ell_1} \in E$, it follows that $z_{\ell_1}'' \notin E$. As $d_i^{\ell_1} \in B \setminus V$ and $E$ is an exact RBDS on vertices outside $V$, it follows that $|E \cap N_{G'}(d^{\ell_1}_i)| = 1$ for all $i\in[k]$. Since $z_{\ell_1}'' \notin E$, it thereby follows that $|E \cap \{u^{\ell_1}_{i,j}\mid  j \in [m_R]\}| = 1$ for all $i \in [k]$.
\end{claimproof}

\begin{claim}\label{claim:gadgets_not_usable}
For any \SERBDS $E$ of $G'$ wrt.~$V$, there exists an index $\ell_2 \in [\sqrt t]$ such that $E \cap V(c^{\ell_2}_{1,j'}) = \emptyset$ for all $j' \in [m_B]$.
\end{claim}
\begin{claimproof}
By Step~\ref{step:makeSprime}, blue vertex~$s'$ has neighborhood $\{y_{\ell} \mid \ell \in [\sqrt{t}]\}$. Since $s' \notin V$, exactly one of these vertices is contained in~$E$; let this be $y_{\ell_2}$. It is connected to the blue vertex of all gadgets $c^{\ell_2}_{1,j'}$ for $j'\in[m_B]$. Since all red vertices in a gadget~$c^{\ell_2}_{1,j'}$ for~$j' \in [m_B]$ have the blue neighbor~$\B{\ell_2}{1}{j'}$ that is also adjacent to~$y_{\ell_2} \in E$, the red vertices in these gadgets are not present in~$E$, as $\B{\ell_2}{1}{j'}$ has exactly one red neighbor in $E$.
\end{claimproof}

\begin{claim}\label{claim:one_from_each_col}
For any \SERBDS $E$ of $G'$ wrt.~$V$, there exists an index $\ell_2 \in [\sqrt t]$ such that for every $j'\in[m_B]$ at least one of the vertices in $\{v^{\ell_2}_{i,j'} \mid i \in [k]\}$ has a neighbor in $E \cap U$.
\end{claim}
\begin{claimproof}
By Claim \ref{claim:gadgets_not_usable} there exists $\ell_2 \in [\sqrt{t}]$ such that $E \cap V(c^{\ell_2}_{1,j'}) = \emptyset$ for all $j' \in [m_B]$.
Consider an arbitrary~$j' \in [m_B]$. The $k$ vertices in $\{v^{\ell_2}_{i,j'} \mid i \in [k]\}$ are connected to vertices of the $k$ gadgets $c^{\ell_2}_{1,j'},c^{\ell_2}_{2,j'},\ldots,c^{\ell_2}_{k,j'}$, and to some vertices in~$U$. From each gadget, at most one red vertex is in $E$, since the red vertices have a common blue neighbor that is not in $V$. Any red gadget vertex is connected to only one vertex in $V$. Since no vertex of gadget $c^{\ell_2}_{1,j'}$ is in $E$, at most $k-1$ of the vertices in $\{v^{\ell_2}_{i,j'} \mid i \in [k]\}$ have a neighbor in $E \cap C_{\ell_2}$. Consequently, at least one of these vertices has a neighbor in $E \cap U$ for each~$j' \in [m_B]$.
\end{claimproof}
We can now prove that $G'$ and $V$ fulfill requirement \ref{req:SERBDS_RBDS} of the lemma statement.
\begin{claim}\label{claim:OR_right}
If $G'$ has a \SERBDS wrt.~$V$, then some input $X_{\ell_1,\ell_2}$ has a \RBDS of size at most $k$.
\end{claim}
\begin{claimproof}
Assume~$G'$ has a \SERBDS wrt.~$V$, say~$E$. By Claim \ref{claim:one_from_each_col}, there exists $\ell_2 \in [\sqrt{t}]$, such that for every $j'\in[m_B]$ at least one of the vertices in $\{v^{\ell_2}_{i,j'} \mid i \in [k]\}$ has a neighbor in $E \cap U$.
By Claim~\ref{claim:all_from_one_square}, there exists $\ell_1 \in [\sqrt{t}]$ such that for all $\ell \neq \ell_1$ we have $U_\ell \cap E = \emptyset$, so these neighbors lie in $U_{\ell_1}$.

We now construct a \RBDS $E'$ for instance $X_{\ell_1,\ell_2}$. For each~$j \in [m_R]$, add $r_j$ to $E'$ if $E \cap \{ u^{\ell_1}_{i,j}\mid i\in[k] \} \neq \emptyset$. By Claim~\ref{claim:all_from_one_square}, it follows that $E'$ has size at most $k$, as required. It remains to show that every vertex in $B_{\ell_1,\ell_2}$ has a neighbor in $E'$.
If some vertex $b_{j'}$ from $B_{\ell_1,\ell_2}$ does not have a neighbor in $E'$, then none of the vertices $\{v^{\ell_2}_{i,j'}\mid i\in[k]\}$ have a neighbor in $E \cap U_{\ell_1}$. This contradicts our choice of $\ell_2$. Hence~$E'$ is an \RBDS of size at most~$k$ for instance~$X_{\ell_1, \ell_2}$.
\end{claimproof}
Furthermore we show that requirement \ref{req:RBDS_ERBDS} is fulfilled in the following claim.
\begin{claim}\label{claim:OR_left}
If some input instance has a \RBDS of size at most $k$, then $G'$ has an \ERBDS.
\end{claim}
\begin{claimproof}
Suppose instance $X_{\ell_1,\ell_2}$ has a \RBDS $E'$ of size $k$ consisting of vertices $r_{i_1},\ldots,r_{i_k} \subseteq R_{\ell_1, \ell_2}$. We construct an \ERBDS $E$ for $G'$. Start by choosing vertices $u^{\ell_1}_{x,i_x}$ for $x\in[k]$, so for every vertex in $E'$ we pick one vertex in the \ERBDS for $G'$. Add the red vertex $z_{\ell_1}$ and the vertices $z_{\ell}''$ for all $\ell \neq \ell_1$ to $E$. Furthermore, we let the vertex~$y_{\ell_2}$ be in $E$.

To exactly dominate the blue vertices in $V$, we use the gadgets in $C$ as follows. For $\ell \neq \ell_2 \in [\sqrt{t}]$, add red vertex $\A{\ell}{x}{x}{j'}$ of gadget $c^\ell_{x,j'}$ if vertex $v^\ell_{x,j'}$ does not yet have a neighbor in $E$, for $j' \in [m_B]$ and~$x \in [k]$. Else, add vertex $\A{\ell}{k+1}{x}{j'}$ of gadget $c^\ell_{x,j'}$ to $E$, in order to exactly dominate the blue vertex of this gadget.

To exactly dominate the vertices in $V_{\ell_2}$ we apply a similar procedure, except that gadget $c^{\ell_2}_{1,j'}$ cannot be used since its blue vertex~$\B{\ell_2}{1}{j'}$ is already dominated by~$y_{\ell_2}$. Since $E'$ is a \RBDS of instance $X_{\ell_1,\ell_2}$, for each~$j' \in [m_B]$ at least one vertex from set $\{v^{\ell_2}_{i,j'} \mid i \in [k]\}$ has a neighbor in $E\cap U$. As such, the $k-1$ remaining gadgets can be used to each dominate one of the $k-1$ remaining vertices in this set, if they do not already have a neighbor in $E \cap U$. If no red vertex of a gadget $c^{\ell_2}_{x,j'}$ is needed to dominate, we choose vertex $\A{\ell_2}{k+1}{x}{j'}$ of the gadget in~$E$ to dominate the blue vertex in the gadget.

It is straight-forward to verify that this results in an \ERBDS for~$G'$.
\end{claimproof}
From Claims~\ref{claim:OR_right} and~\ref{claim:OR_left} it follows that graph $G'$ has a \SERBDS wrt.~$V$ if and only if at least one of the input instances has a \RBDS of size at most $k$. The graph $G'$ has $\Oh(\sqrt{t} \cdot (m_R+m_B)^3)$ vertices and can be constructed in polynomial time.
\end{proof}
Using the lemma above, we now prove the kernel lower bound for \ERBDS.
\begin{note}
\begin{theorem}\label{thm:exact_RBDS}
\problem{Exact Red-Blue Dominating Set} parameterized by the number of vertices $n$ does not have a generalized kernel of size $\Oh(n^{2-\varepsilon})$, unless \containment.
\end{theorem}
\begin{proof}
We will prove this result by giving a degree-$2$ cross-composition from \RBDS to \ERBDS. We start by giving a polynomial equivalence relation \eqvr on inputs of \RBDS. Let two instances of \RBDS be equivalent under \eqvr if they have the same number of red vertices $m_R$, the same number of blue vertices $m_B$, and the same maximum size~$k$ of a \RBDS. It is easy to check that \eqvr is a polynomial equivalence relation.

Assume we are given $t$ instances of \RBDS, labeled $X_{i,j}$ for $i,j\in[\sqrt{t}]$, from the same equivalence class of \eqvr. If the number of instances given is not a square, we duplicate one of the input instances until a square number is reached. Since this changes the number of inputs by at most a factor four, this does not influence the cross-composition.  Call the number of red vertices in every instance $m_R$, the number of blue vertices $m_B$, and the required size of the dominating set $k$. By Lemma \ref{lem:cross_composition_erbds}, we can in polynomial time construct graph $G'$ such that
\begin{itemize}
\item $|V(G)| \leq \sqrt{t}\cdot\text{poly}(m_B+m_R)$ and
\item $G'$ has an \ERBDS if and only if at least one input instance has a \RBDS. This follows from requirements \ref{req:SERBDS_RBDS} and \ref{req:RBDS_ERBDS} from Lemma \ref{lem:cross_composition_erbds}, and the fact that any \ERBDS is also a \SERBDS.
\end{itemize}
Thereby we have given a degree-$2$ cross-composition and the lower bound follows from Theorem \ref{thm:cross_composition_LB}.
\end{proof}
\end{note}

Using Theorem~\ref{thm:exact_RBDS} we provide lower bounds for constraint satisfaction problems.
It is easy to give a linear parameter transformation from \ERBDS to both \linearrootCSP and \ExactSAT, by introducing a variable for each red vertex and adding a constraint for each blue vertex such that exactly one of its neighbors is chosen in any assignment. This results in the following corollary.

\begin{corollary}\label{cor:LB:1-poly-root-csp}
The problems \ExactSAT and \linearrootCSP over~$\mathbb{Q}$, parameterized by the number of variables~$n$, do not have a generalized kernel of size $\Oh(n^{2-\varepsilon})$ for any~$\varepsilon > 0$, unless \containment.
\end{corollary}

\subsection{Polynomial root CSP modulo an integer}

In order to also establish a lower bound for \linearrootCSP over the integers modulo $m$, we will need the following lemma. It allows us to enforce a linear equality constraint over~$\mathbb{Q}$ using constraints over~$\mathbb{Z}/m\mathbb{Z}$, through the use of auxiliary $0/1$-dummy variables. Since~$\sum_i x_i = 1$ implies~$\sum_i x_i \equiv_m 1$, the nontrivial part is to add extra constraints which, together with~$\sum_i x_i \equiv_m 1$, also imply~$\sum_i x_i = 1$.

\begin{lemma}\label{lem:rewrite_to_mod_m}
Let $m \geq 3$ be an integer. Given a linear equality $\sum_{i \in [N]} x_i = 1$ over $\mathbb{Q}$, there exists a system $S$ of linear equalities over $\mathbb{Z}/m\mathbb{Z}$ using the variables $\{x_i\mid i\in [N]\}$ and at most $4N$ additional variables, such that
\begin{enumerate}
\item Any $0/1$-solution to the system $S$ sets exactly one of the variables $\{x_i\mid i\in[N]\}$ to $1$,\label{rewrite:solution:sets:one}
\item any assignment to $\{x_i\mid i\in[N]\}$ setting exactly one variable $x_i$ to $1$ can be extended to a $0/1$-solution of $S$, and\label{rewrite:condition:extend:to:solution}
\item $S$ can be constructed in polynomial time.\label{rewrite:condition:polytime}
\end{enumerate}
\end{lemma}
\begin{proof}
Given the linear equality $\sum_{i \in [N]} x_i = 1$, first of all add the equation
\[\sum_{i \in [N]} x_i \equiv_m 1\] to $S$.
 Any choice of $x_1,\ldots, x_{N}$ satisfying $\sum_{i \in [N]} x_i = 1$ also satisfies the equality modulo $m$. Furthermore, any $0/1$-assignment of $x_1,\ldots, x_{N}$ satisfying $\sum_{i \in [N]} x_i \equiv_m 1$ ensures that \emph{at least} one variable $x_i$ is set to $1$.

To ensure that \emph{at most} one of these variables is set to $1$, we add additional constraints in the following way. Construct a complete binary tree with $N' := 2^{\lceil\log N \rceil}$ leaves, implying $N \leq N' < 2N$. Identify the first $N$ leaves with variables $x_1,\ldots,x_N$ and introduce dummy variables for all other vertices. For every non-leaf $d$ in the tree with children $d_\ell$ and $d_r$, each corresponding to a unique variable, add the equation
\[d_\ell + d_r \equiv_m d.\]
It is clear that this construction can be done in polynomial time, thus Property~\ref{rewrite:condition:polytime} holds. To show that Properties~\ref{rewrite:solution:sets:one} and~\ref{rewrite:condition:extend:to:solution} hold, we prove the following claim.
\begin{note}
\begin{claim}\label{claim:at_most_one_leaf}
Let a $0/1$-assignment satisfying all equalities in $S$ be given. The value assigned to any variable $x$ corresponds to the number of leaves in the subtree rooted in $x$ that are assigned value~$1$.
\end{claim}
\begin{claimproof}
We prove this by induction on the height of the tree rooted in $x$. If $x$ is a leaf, the result is obvious. Suppose the tree has height larger than one and let $x_\ell$ and $x_r$ be the left and right child of $x$. By the induction hypothesis, the values of $x_\ell$ and $x_r$ correspond to the number of leaves in the left (respectively, right) subtree that were assigned $1$. Since $x_\ell,x_r \in \{0,1\}$ and $x \equiv_m x_\ell + x_r$ with $m > 2$, the result follows.
\end{claimproof}

Suppose we are given any $0/1$-assignment satisfying all equalities in $S$. Hence the variable corresponding to the root $r$ of the binary tree has value~$0$ or~$1$. By Claim \ref{claim:at_most_one_leaf}, it follows that the number of leaves (and thus the number of variables in $\{x_1,\ldots,x_N\}$) that are assigned the value $1$ is at most one. As we have seen earlier, at least one variable $x_i$ is set to $1$, to fulfill $\sum_{i \in [N]} x_i \equiv_m 1$, and thus $\sum_{i \in [N]} x_i = 1$. Hence Property~\ref{rewrite:solution:sets:one} holds.
\end{note}

Given a $0/1$-assignment to $x_1,\ldots,x_N$ such that $\sum_{i \in [N]} x_i = 1$, it can be extended to a satisfying assignment of $S$ by setting all dummy leaves to $0$. For every other dummy vertex, let its value be the number of variables corresponding to leaves in its subtree, that are set to $1$. Note that this number is always either~$0$ or~$1$ since there is only one leaf whose corresponding variable is set to~$1$. Therefore Property~\ref{rewrite:condition:extend:to:solution} holds as well.
\end{proof}

For $m=2$, an input to the problem \linearrootCSP over the integers mod $m$ only consists of linear equations over the two-element field $\{0,1\}$ and is thus polynomial time solvable by Schaefer's dichotomy theorem \cite[Theorem 2.1]{Schaefer78}. For larger moduli, we use Lemma~\ref{lem:rewrite_to_mod_m} to prove the following result.

\begin{note}
\begin{theorem}\label{cor:LB:1-poly-root-csp-modm}
Let $m \geq 3$ be an integer. The problem \linearrootCSP over~$\mathbb{Z}/m\mathbb{Z}$, parameterized by the number of variables~$n$, does not have a generalized kernel of size $\Oh(n^{2-\varepsilon})$ for any~$\varepsilon > 0$, unless \containment.
\end{theorem}
\begin{proof}
We will use the graph constructed in Lemma \ref{lem:cross_composition_erbds}, by transforming the constructed instance $G'$ of (\textsc{semi})-\ERBDS of size $\Oh(\sqrt{t}\cdot\text{poly}(m_R+m_B))$ to an instance $I$ of  \linearrootCSP over~$\mathbb{Z}/m\mathbb{Z}$ with $\Oh(\sqrt{t}\cdot\text{poly}(m_R + m_B))$ variables. In this way we obtain a degree-$2$ cross-composition from \RBDS to \linearrootCSP over~$\mathbb{Z}/m\mathbb{Z}$, proving the lower bound.

Suppose we are given $t$ instances of \RBDS, such that $\sqrt{t}$ is integer and such that every instance has $m_B$ blue vertices and $m_R$ red vertices and asks for a \RBDS of size $k\leq m_R$. This can be assumed by choosing an appropriate polynomial equivalence relation. Apply Lemma \ref{lem:cross_composition_erbds} to obtain graph $G'$ and $V \subseteq V(G')$. By requirements \ref{req:RBDS_ERBDS} and \ref{req:SERBDS_RBDS} of Lemma \ref{lem:cross_composition_erbds}, it is sufficient to ensure that $G'$ has a \SERBDS with respect to $V$ if $I$ is satisfiable, and that~$I$ is satisfiable if $G$ has an \ERBDS, to obtain the cross-composition.

Recall that a \SERBDS of $G'$ with respect to $V$ contains at least one neighbor of each blue vertex, and contains \emph{exactly} one neighbor of each blue vertex in $V(G') \setminus V$.

 First of all introduce a variable~$v_r$ for every red vertex~$r$ in $G'$. For every blue vertex $b$, we add the following equation to ensure that it has at least one neighbor in the \SERBDS:
\begin{equation}\label{eqn:at_leat_one_nb}
\sum_{r \in N_{G'}(b)} v_r \equiv_m 1.
\end{equation}

For every blue vertex $b \notin V$, we add a number of linear equations that ensure $b$ has exactly one neighbor in a \SERBDS, using at most $4\cdot|N_{G'}(b)|$ additional variables. This is done by applying Lemma \ref{lem:rewrite_to_mod_m} to the equation $\sum_{r \in N_{G'}(b)} v_r = 1$.

This completes the construction. If~$G'$ has an \ERBDS, then $I$ can be satisfied by setting the variables corresponding to the \ERBDS to $1$ and all other variables corresponding to vertices to $0$. The dummy variables can then be chosen in such a way that all equations are satisfied according to Lemma \ref{lem:rewrite_to_mod_m}.

For the opposite direction, suppose $I$ has a satisfying assignment. Define set $Y$ to contain the vertices whose corresponding variable is set to $1$. From Equation~(\ref{eqn:at_leat_one_nb}) it follows that every blue vertex has at least one neighbor in the set $Y$. Furthermore every blue vertex not in $V$ has exactly one neighbor in  $Y$ by Lemma \ref{lem:rewrite_to_mod_m}.  It follows that $Y$ is a \SERBDS of $G'$.

It remains to bound the number of used variables. The key idea is that we have only few variables outside of $V$ whose corresponding vertex has a large neighborhood, and for which the number of dummy variables added depends on $\sqrt{t}$. Furthermore there are many variables whose corresponding vertices have small neighborhoods, with size depending only on  $m_B + m_R$. 
Note that the degree of any vertex, and the total number of vertices, is bounded by the order of the graph~$\Oh (\sqrt{t} \cdot (m_R + m_B)^3)$.

For every blue vertex in $V(G)\setminus V$ with a degree larger than $m_R + k + 2$ we add $\Oh(\sqrt{t} (m_B+m_R)^3)$ dummy variables. By requirement \ref{req:degrees} of Lemma \ref{lem:cross_composition_erbds}, there are at most $2$ such vertices. Furthermore for any vertex with a degree smaller than $m_R + k + 2$ we add $\Oh(m_R + k)$ dummy vertices. This together results in using $\Oh(\sqrt{t}\cdot \text{poly}((m_{B} + m_{R})))$ variables, which is properly bounded for a degree-$2$ cross-composition.
\end{proof}
\end{note}

We now generalize this result to polynomial equalities of higher degree.
\begin{theorem} \label{thm:drootcspmod:lb}
Let $m \geq 2$, and $d \geq 2$ be integers. The problem \rootCSP over~$\mathbb{Z}/m\mathbb{Z}$ and \rootCSP over $\mathbb{Q}$ parameterized by the number of variables $n$ do not have a generalized kernel of size $\Oh(n^{d+1-\varepsilon})$ for any~$\varepsilon > 0$, unless \containment.
\end{theorem}
\begin{proof}\begin{note}
We will only provide the proof over $\mathbb{Z}/m\mathbb{Z}$, the result for \rootCSP over $\mathbb{Q}$ can be obtained in the same way (using the same equations without the moduli). Let $m \geq 2, d \geq 2$ be given. The result will be proven by a degree-$(d+1)$ cross-composition from \RBDS, using Lemma \ref{lem:cross_composition_erbds}. Suppose we are given $t = r^{d+1}$ instances of \RBDS, all having~$m_R$ red vertices,~$m_B$ blue vertices, and the same target size~$k$. By a similar padding argument as before, we may assume~$r$ is an integer. Split the inputs into $r^{d-1}$ groups of size~$r^2$ each and apply the algorithm given by Lemma \ref{lem:cross_composition_erbds} to each group. We obtain $r^{d-1}$ instances of \problem{(semi)-\ERBDS} with $\Oh(r \cdot (m_R+m_B)^3)$ vertices each, such that the answer to each composed instance is the logical \OR of the answers to the \RBDS instances in its group. Label the instances resulting from the group compositions~$X_{i_1,\ldots, i_{d-1}}$ with $i_1,\ldots,i_{d-1} \in [r]$. Let instance $X_{i_1,\ldots,i_{d-1}}$ have graph $G_{i_1,\ldots,i_{d-1}}$ and let the set on which the \RBDS is not required to be exact be~$V_{i_1, \ldots, i_{d-1}}$. All produced graphs have the same number of red and blue vertices; let the number of red vertices in each graph be~$N \leq |V(G_{i_1, \ldots, i_{d-1}})| \leq \Oh(r \cdot (m_R + m_B)^3)$. We create~$N$ new variables and identify each red vertex $x$ with one variable $v_x$. It is essential for the remaining part of this proof that vertices from the produced \problem{(semi)-\ERBDS} instances that had the same label are mapped to the same new variable and vice versa. Since the set~$V_{i_1, \ldots, v_{d-1}}$ that is produced by Lemma~\ref{lem:cross_composition_erbds} does not depend on the structure of the input graphs, only on their size, all produced graphs have the same labeled vertices in the set~$V_{i_1, \ldots, i_{d-1}}$. Hence we can treat it as a single set~$V$ of vertex labels. Create an instance for \rootCSP as follows.

\begin{enumerate}
\item Add sets $Y_1,\ldots, Y_{d-1}$ of~$r$ variables each, where $Y_i := \{y^i_j\mid j\in[r]\}$. Add the requirement
$\sum_{j \in [r]} y^i_j \equiv_m 1$ to $L'$ for each $i \in [d-1]$.
\item \label{step:add_eq_mod} Consider each graph $G_{i_1,\ldots, i_{d-1}}$ for~$i_1, \ldots, i_{d-1} \in [r]$. For each blue vertex $b$ in this instance, add the following equation to $L'$:
\begin{equation}\label{eqn:nb-hood}    \left( \sum_{x \in N_{G_{i_1, \ldots, i_{d-1}}}(b)} v_x\right)\cdot\prod_{z \in [d-1]} y^z_{i_z}  \equiv_m \prod_{z \in [d-1]} y^z_{i_z}.\end{equation}
Furthermore, if $b$ is not an element of $V$, then for every pair of distinct vertices $x,x' \in N(b)$ add the following constraint to~$L'$:
\begin{equation}\label{eqn:mult_two_nb_is_zero}v_x\cdot v_{x'} \equiv_m 0.\end{equation}
Note that, by Observation~\ref{obs:same_nbhood}, the neighborhood of a blue vertex~$b \not \in V$ does not depend on the graphs to which Lemma~\ref{lem:cross_composition_erbds} is applied, but only on the number of red and blue vertices and the target size of the \RBDS. As these are identical for all applications of the lemma, it does not matter in which of the graphs we evaluate~$N(b)$ when finding relevant pairs~$x,x'$.
\end{enumerate}
The polynomial equalities have degree~$\leq d$ as $d$ is at least two. The number of variables, which is the parameter of the CSP, is suitably bounded for a degree-$(d+1)$ cross-composition:
\[N + (d-1) \cdot r \in \Oh(r \cdot (d + (m_R + m_B)^3)) = \Oh(t^{1/(d+1)} (m_R + m_B)^3).\]

As the construction can easily be performed in polynomial time, it remains to show that the constraints in~$L'$ can be satisfied if and only if one of the input instances of \RBDS has a solution of size~$k$. First assume that some input instance or \RBDS indeed has a solution of size~$k$. Consider the indices~$i_1, \ldots, i_{d-1}$ of the group containing the satisfiable \RBDS instance. Then Lemma~\ref{lem:cross_composition_erbds} ensures that~$G_{i_1, \ldots, i_{d-1}}$ has an \ERBDS. Set the variables corresponding to vertices in the \ERBDS of $G_{i_1,\ldots,i_{d-1}}$ to $1$ and the others to~$0$. Furthermore, set variables $y^z_{i_z}$ for $z\in[d-1]$ to~$1$. Set all other variables to~$0$. Thereby the sum of variables in each set $Y_i$ is~$1$, as required. Furthermore, each equation defined by (\ref{eqn:nb-hood}) is satisfied in the following way. If it was defined for $X_{i_1,\ldots,i_{d-1}}$, it is satisfied since the large summation equals one (exactly one neighbor is in the exact dominating set) and the product term is one on both sides. Equations belonging to any other instance are trivially satisfied since their term~$\prod_z y^z_{\cdot}$ is zero on both sides. It remains to show that the equations defined by (\ref{eqn:mult_two_nb_is_zero}) are satisfied. This is follows from Observation \ref{obs:same_nbhood} and the fact that an \ERBDS contains at most one neighbor of each blue vertex.

For the reverse direction, suppose the constraints in $L'$ are satisfied by some $0/1$-assignment to the variables. Then from each set $Y_i$ with $i \in [d-1]$, at least one variable is set to $1$. So suppose variables $y^z_{i_z}$ are set to $1$ for $z \in [d-1], i_z \in [r]$. We show instance $X_{i_1,\ldots,i_{d-1}}$ has a \SERBDS wrt.~$V$ consisting of the vertices whose corresponding variable is set to $1$. Since the product~$\prod_{z \in [d-1]} y^z_{i_z}$ is~$1$ on both sides of the equations defined by (\ref{eqn:nb-hood}) for~$G_{i_1, \ldots, i_{d-1}}$, for each blue vertex~$b$ in the graph we have:
\begin{align*}
\sum_{x \in N_{G_{i_1, \ldots, i_{d-1}}}(b)} v_x \equiv_m 1
\end{align*}
implying all blue vertices have at least one neighbor in the \SERBDS. Furthermore if $x \notin V$, we know that it has at most one neighbor in the \SERBDS since the multiplication of any two of its neighbors yields zero by (\ref{eqn:mult_two_nb_is_zero}). Hence~$G_{i_1, \ldots, i_{d-1}}$ has an \SERBDS wrt.~$V$. By Lemma~\ref{lem:cross_composition_erbds}, this implies the group of \RBDS instances from which it was constructed contained a satisfiable instance. Hence there was a \yes-instance among the inputs of the cross-composition.
\end{note}
\end{proof}

Observe that the polynomials constructed in Theorem~\ref{thm:drootcspmod:lb} have a simple form: each polynomial is a product of $(d-1)$ $Y$-variables multiplied by a sum of variables corresponding to red vertices, or simply a multiplication of two variables corresponding to red vertices. Each polynomial can therefore be encoded in~$\Ohtilde(n)$ bits, where~$n$ is the number of variables in the constructed CSP. The sparsification of Theorem~\ref{thm:kernel} therefore encodes such instances in~$\Ohtilde(n^{d+1})$ bits. The lower bound shows that this is optimal up to~$n^{o(1)}$ factors.

\subsection{Polynomial non-root CSP} \label{sec:nonroot:lb}

We start our lower bound discussion for \nonrootCSP by considering polynomials over the rationals. Using existing kernel lower bounds for \textsc{CNF-Satisfiability} parameterized by the number of variables, we first show that \linearNonrootCSP over~$\mathbb{Q}$ does not have a generalized kernel of size bounded by any polynomial in~$n$, unless \containment.

\begin{theorem} \label{thm:nonroot:q:lb}
\linearNonrootCSP over~$\mathbb{Q}$ parameterized by the number of variables $n$ does not have a generalized kernel of polynomial size unless \containment.
\end{theorem}
\begin{proof}
We present a linear-parameter transformation from \textsc{CNF-Satisfiability} with unbounded clause length parameterized by the number of variables. Existing results~\cite{SatisfiabilityDell14,FortnowS11} imply that this problem does not have a generalized kernel of polynomial size. The linear-parameter transformation will transfer this lower bound to \linearNonrootCSP over~$\mathbb{Q}$.

A clause in conjunctive normal form can directly be translated into a non-root constraint of a degree-1 polynomial over~$\mathbb{Q}$. For example, the clause $(x_1 \vee \neg x_3 \vee x_4)$ is satisfied by a $0/1$-assignment if and only if~$x_1 + (1 - x_3) + x_4 \neq 0$ over~$\mathbb{Q}$. More generally, a clause~$(x_{i_1} \vee \ldots \vee x_{i_k} \vee \neg x_{i_{k+1}} \vee \ldots \vee \neg x_{i_\ell})$ translates into the constraint~$(\sum _{j=1}^k x_{i_j}) + (\sum _{j=k+1}^\ell (1-x_{i_j})) \neq 0$. Hence the system of inequalities derived by transforming all clauses in a CNF formula is satisfiable if and only if the formula is. As the number of variables is preserved by this transformation, the theorem follows.
\end{proof}

We now turn our attention to \nonrootCSP over finite rings and fields. In Theorem \ref{thm:UB:non-root} we provided a kernel for \nonrootCSP over \primefield for primes~$p$. It is natural to ask whether similar results can be obtained when working with polynomials modulo an arbitrary integer~$m$. When~$m$ is composite, our kernelization fails. We can show that this is not a shortcoming of our proof strategy, but a necessity due to the fact that constraints expressed by degree-$d$ polynomials modulo composite numbers can model more complex constraints than degree-$d$ polynomials modulo a prime. For example, it is known (cf.~\cite[\S 2]{Barrington1992some}) that there is a degree-$3$ polynomial~$f$ over the integers modulo~$6$ which represents a logical \OR of size~$27$ in the following way:
\begin{equation} \label{eq:or}
f(x_1,\ldots,x_{27}) \not\equiv_6 0 \Leftrightarrow (x_1\vee\ldots\vee x_{27}).
\end{equation}
By this expressibility of a size-$27$ \OR by a polynomial of degree $3$ over $\mathbb{Z}/6\mathbb{Z}$ using the same variables, one easily constructs a linear-parameter transformation from \problem{$27$-cnf-sat} to \problem{$3$-Polynomial non-root CSP} over~$\mathbb{Z}/6\mathbb{Z}$ by mimicking the proof of Theorem~\ref{thm:nonroot:q:lb}. Since \problem{$27$-cnf-sat} does not have a kernel of size~$\Oh(n^{27-\epsilon})$ for any~$\epsilon > 0$ unless \containment (Theorem~\ref{thm:lower_bound:CNF}), this linear-parameter transformation rules out kernels of size~$\Oh(n^{27-\epsilon})$ for \problem{$3$-Polynomial non-root CSP} over~$\mathbb{Z}/6\mathbb{Z}$ under the same conditions. Plugging in the degree of $3$ and modulus $6$ into the bound of Theorem~\ref{thm:UB:non-root} would give a reduction to $\Oh(n^{3\cdot(6-1)}) = \Oh(n^{15})$ constraints and would contradict the lower bound. The example therefore shows that the problem is more complex for composite moduli: the bound for the prime case cannot be matched. In particular, we will see that the exponent in the kernel size may depend super-linearly on the degree~$d$ of the CSP. For general non-primes, we give a lower bound using a construction by Bhowmick \etal~\cite{Bhowmick2015Nonclassical} of low-degree polynomials representing \OR in the sense of Equation~\ref{eq:or}.

\begin{theorem}
\label{thm:nonroot:composite:lb}
Let $m$ be a non-prime with a prime factorization consisting of $r$ distinct primes, such that $m = \prod_{i \in [r]} p_i$. Let~$d$ be an even integer. Then \nonrootCSP over~$\mathbb{Z}/m\mathbb{Z}$ parameterized by the number of variables $n$ does not have a generalized kernel of size $\Oh(n^{(d/2)^r-\varepsilon})$ for any $\varepsilon  > 0$, unless \containment.\footnote{This theorem corrects a statement in the extended abstract of this work, in which a lower bound of $\Omega(n^{(d^r)/2-\varepsilon})$ was erroneously claimed.}
\end{theorem}
\begin{proof}
For any integer~$N \geq 1$, Bhowmick \etal \cite[Appendix A]{Bhowmick2015Nonclassical} provide a way to construct a polynomial $f$ of degree  $2\lceil N^{1/r}\rceil $ such that for all $x_1,\ldots,x_N \in \{0,1\}$,
\begin{equation}
f(x_1,\ldots,x_N) \not\equiv_m 0 \Leftrightarrow (x_1\vee\dots\vee x_N).
\end{equation}
This implies that for even values of~$d$ and~$N = (d/2)^r$, we can find a polynomial~$f$ of degree~$d$ satisfying the above equation. As such, \nonrootCSP can express a logical \OR of size $(d/2)^r$ without introducing auxiliary variables. As in the proof of Theorem~\ref{thm:nonroot:q:lb}, this gives a linear-parameter transformation from \problem{$(d/2)^r$-cnf-sat} to \nonrootCSP. By Theorem~\ref{thm:lower_bound:CNF}, the latter problem does not have a generalized kernel of size $\Oh(n^{(d/2)^r-\varepsilon})$ for any $\varepsilon  > 0$, unless \containment. Hence the same lower bound applies to the CSP.
\end{proof}

In case $m$ does not have a prime factorization in which all primes are distinct, it is possible to obtain weaker a lower bound using a result by Barrington \etal \cite{Barrington1994representing}, which proves that there exists a polynomial of degree $\Oh(\ell N^{1/r})$ that represents a logical \OR when taken modulo $m$. Here $\ell$ is the largest prime factor of~$m$. For prime moduli, the following result provides a lower bound almost matching the upper bound in Section \ref{sec:UB:non-root}.

\begin{theorem}
\label{thm:nonroot:prime:lb}
Let $p$ be a prime. Then \nonrootCSP over~$\mathbb{Z}/p\mathbb{Z}$ parameterized by the number of variables $n$ does not have a generalized kernel of size $\Oh(n^{d(p-1)-\varepsilon})$ for any $\varepsilon > 0$, unless \containment.
\end{theorem}
\begin{proof}
We use a linear-parameter transformation from \problem{$d(p-1)$-cnf-sat}. We proceed similarly as in the proof of Theorem~\ref{thm:nonroot:composite:lb}. It is known (cf.~\cite[Theorem 24]{Beigel93}) that for each prime~$p$ and integer~$d$, there is a polynomial~$f$ of degree~$d$ modulo~$p$, such that for any~$x_1, \ldots, x_{d(p-1)} \in \{0,1\}$ we have:
\[f(x_1,\ldots,x_{d(p-1)}) \not\equiv_p 0 \Leftrightarrow (x_1 \vee x_2 \vee \dots\vee x_{d(p-1)}).\]
This allows the linear-parameter transformation to be carried out as in Theorem~\ref{thm:nonroot:composite:lb}.
\end{proof}

\section{Conclusion} \label{sect:conclusion}

We have given upper and lower bounds on the kernelization complexity of binary CSPs that can be represented by polynomial (in)equalities, obtaining tight sparsification bounds in several cases. Our main conceptual contribution is to analyze constraints on binary variables based on the minimum degree of multivariate polynomials whose roots, or non-roots, capture the satisfying assignments. The ultimate goal of this line of research is to characterize the optimal sparsification size of a binary CSP based on easily accessible properties of the constraint language. To reach this goal, several significant hurdles have to be overcome.

For \nonrootCSP over the integers modulo~$6$, we do not know of any way to reduce the number of constraints to polynomial in~$n$. This difficulty is connected to longstanding questions regarding the minimum degree of a multivariate polynomial modulo~$6$ that represents the \OR-function of~$n$ variables in the sense of Equation~\ref{eq:or}. As exploited in the construction of Theorem~\ref{thm:nonroot:composite:lb}, if the \OR-function with~$g(d)$ inputs can be represented by polynomials of degree~$d$, then \nonrootCSP cannot be compressed to size~$\Oh(n^{g(d)-\varepsilon})$ unless \containment. By contraposition, a kernelization with size bound~$\Ohtilde(n^{h(d)})$ implies a lower bound of~$h^{-1}(d)$ on the degree of a polynomial representing an \OR of arity~$h(d)$, assuming \ncontainment. Kernel bounds where~$h(d)$ is polynomially bounded in~$d$, would therefore establish lower bounds of the form~$\Omega(n^\alpha)$ on the degree of polynomials representing an $n$-variable \OR modulo 6, for some~$\alpha>0$. However, the current-best degree lower bound~\cite{TardosB98} is only~$\Omega(\log n)$, which has not been improved in nearly two decades (cf.~\cite[\S 1.4]{Bhowmick2015Nonclassical}).

When it comes to CSPs whose constraints are of the form ``the number of satisfied literals in the clause belongs to set~$S$'', many cases remain unsolved. We can prove (see Appendix \ref{ap:primesat}) that for constraints of the form ``the number of satisfied literals is a prime number'', no generalized kernel of size polynomial in~$n$ exists unless \containment. On the other hand, Corollary~\ref{cor:S_SAT} gives good compressions for problems of the type ``the number of satisfied literals in the clause is a multiple of three''. Is sparsification possible when a constraint requires the number of satisfied literals to be a square, for example?

A simple example of a CSP whose kernelization complexity is currently unclear has constraints of the form ``the number of satisfied literals is one or two, modulo six''. The approach of Theorem \ref{thm:kernel} fails, since there is no polynomial modulo six with root set~$\{1,2\}$.

Finally, we mention that all our results extend to the setting of min-ones and max-ones CSPs, in which one has to find a satisfying assignment that sets at least, or at most, a given number of variables to true. For example, our results easily imply that \textsc{Exact Hitting Set} parameterized by the number of variables~$n$ has a sparsification of size~$\Oh(n^2)$, which cannot be improved to~$\Oh(n^{2-\varepsilon})$ unless \containment.

\subparagraph*{Acknowledgements}
We are grateful to Fedor Petrov for suggesting Lemma~\ref{lemma:spanningset:modm}.

\bibliographystyle{plainurl}
\bibliography{report}
\clearpage
\appendix
\section{Prime SAT}\label{ap:primesat}



\medskip

In this appendix we consider the following variant of the satisfiability problem, in which a clause is satisfied if the number of satisfied literals is a prime.

\defparproblem{\problem{Prime-Sat}}
{A set of clauses $\mathcal{C}$ over variables~$V := \{x_1, \ldots, x_n\}$. Each clause is a set of distinct literals of the form~$x_i$ or~$\neg x_i$.}
{The number of variables $n$}
{Does there exist a truth assignment for the variables~$V$ such that the number of satisfied literals in clause $i$ is a prime for all $i$?}

\begin{theorem} \label{thm:primesat:lb}
\problem{Prime-Sat} parameterized by the number of variables does not have a polynomial kernel unless \containment.
\end{theorem}

\begin{proof}
We show the non-existence of a polynomial kernel by giving a linear parameter transformation from \dSAT for any~$d$, which establishes the claimed lower bound by Theorem~\ref{thm:lower_bound:CNF}. So fix an integer~$d$ and let an instance $\mathcal{F}$ of \dSAT be given. 

It is proven in \cite{green2008primes} that the primes contain arbitrarily long arithmetic progressions; hence there is an arithmetic progression of length at least~$d$ among the primes. Let~$\{c + i \cdot b \mid i \in \mathbb{N}_0 \wedge i < d\}$ be an arithmetic progression of~$d$ primes. We claim that this arithmetic progression has a finite length: there is some integer~$j$ such that~$c + j \cdot b$ is not prime. To see this, note that $b+1$ divides $c + (c + b + 1) \cdot b = (c+b)(b+1)$, which bounds the length of this progression. Hence we can choose~$a \geq c$ such that~$\{a + i \cdot b \mid i \in \mathbb{N}_0 \wedge i < d\}$ is a set of~$d$ primes, while~$a + d \cdot b$ is not prime. Using~$a$ and~$b$, we transform the instance~$\mathcal{F}$ of \dSAT into an equivalent instance of \problem{Prime-Sat}, as follows.

For each clause $C_i = (\ell_1,\ldots,\ell_d)$ in $\mathcal{F}$, we add a clause~$C'_i$
\[(\underbrace{1,\dots,1}_{a \text{ copies}},\underbrace{\neg\ell_1,\dots,\neg\ell_1}_{b \text{ copies}},\ldots,\underbrace{\neg\ell_d,\dots,\neg\ell_d}_{b \text{ copies}})\]
to the \problem{Prime-Sat} instance $\mathcal{F'}$. If an assignment of the variables in~$\mathcal{F}$ satisfies~$i$ literals of~$C_i$, then the corresponding clause~$C'_i$ will have~$a + (d-i) \cdot b$ satisfied literals. By our choice of~$a$ and~$b$, this number of prime if and only if~$i > 0$. Hence~$C_i$ is satisfied for \dSAT exactly when~$C'_i$ is satisfied for \problem{Prime-Sat}.

So far, the construction uses multiple occurrences of the same variable, and also uses the constant $1$. Formally, this is not allowed in the definition of \problem{Prime-Sat}. We resolve this issue by replacing the constants by $a$ new variables $T_1,\ldots,T_a$. These can be forced to \true by adding clauses $(T_i,T_{i+1})$ for $i \in [a-1]$, since two is a prime number while zero and one are not. For each variable $x$ we add $b$ distinct copies $x_1,\dots,x_b$ and require them to be equal with clauses $(T_1,T_2,T_3,x_i,x_{i+1})$ for $i \in [b-1]$. Since three and five are primes while four is not, while all~$T_i$ are forced to true by the earlier part of the construction, this clause is only satisfied when the two copies are both \true, or both \false. In this way we eliminate the need for repeated variables in the clauses of~$\mathcal{F}'$. As the number of variables increases by a constant factor depending only on~$d$, which is fixed, this yields a valid linear-parameter transformation for each~$d$.
\end{proof}

\end{document}